\documentclass{llncs}

\usepackage{amssymb,stmaryrd,mathrsfs,dsfont,mathtools,multirow,multicol}

\newtheorem{notation}[theorem]{Notation}

\input xy
\xyoption{all}

\usepackage{tikz}
\tikzstyle{vertex}=[circle, draw, inner sep=0pt, minimum size=4pt, fill=black]
\newcommand{\vertex}{\node[vertex]}

\begin{document}

\frontmatter         

\pagestyle{headings}  

\mainmatter            

\title{On the Permutation-Representation Number of Bipartite Graphs using Neighborhood Graphs}

\titlerunning{The Permutation-Representation Number of Bipartite Graphs} 

\author{Khyodeno Mozhui \and K. V. Krishna} 

\authorrunning{Khyodeno Mozhui \and K. V. Krishna} 

\institute{Indian Institute of Technology Guwahati, Guwahati, India\\
	\email{k.mozhui@iitg.ac.in};
	\email{kvk@iitg.ac.in}}

\maketitle             

\begin{abstract}
	The problems of determining the permutation-representation number (\textit{prn}) and the representation number of bipartite graphs are open in the literature. Moreover, the decision problem corresponding to determination of the \textit{prn} of a bipartite graph is NP-complete. However, these numbers were established for certain subclasses of bipartite graphs, e.g., for crown graphs. Further, it was conjectured that the crown graphs have the highest representation number among the bipartite graphs. In this work, first, we reconcile the relation between the \textit{prn} of a comparability graph and the dimension of its induced poset and review the upper bounds on the \textit{prn} of bipartite graphs. Then, we study the \textit{prn} of bipartite graphs using the notion called neighborhood graphs. This approach substantiates the aforesaid conjecture and gives us theoretical evidence. In this connection, we devise a polynomial-time procedure to construct a word that represents a given bipartite graph permutationally. Accordingly, we provide a better upper bound for the \textit{prn} of bipartite graphs. Further, we construct a class of bipartite graphs, viz., extended crown graphs, defined over posets and investigate its \textit{prn} using the neighborhood graphs. 
\end{abstract}


\section{Introduction}
The class of word-representable graphs, which includes a wide range of graphs such as comparability graphs, circle graphs, and 3-colorable graphs, is an important class of graphs with intriguing properties. The monograph \cite{kitaev15mono} provides a comprehensive account of word-representable graphs, their connections to other contexts, and contributions to the topic. Finding the representation number of an arbitrary word-representable graph is a hard problem, and many authors have focused on determining the numbers for various subclasses. 

The class of comparability graphs --  the graphs which admit  transitive orientations -- is an important subclass of word-representable graphs. Indeed, the class of comparability graphs is precisely the class of permutationally representable graphs, each of which can be represented by a concatenation of permutations of its vertices \cite{kitaev08order}. The minimum number of such permutations of vertices of a comparability graph is called the permutation-representation number (\textit{prn}) of the graph. Every comparability graph induces a poset based on one of its transitive orientations. The \textit{prn} of a comparability graph is precisely the dimension of the corresponding induced poset (cf. Section \ref{prn=dim}). For $ 3 \le k \le \lceil \frac{n}{2} \rceil$, it is known that determining whether a poset with $n$ elements has dimension $k$ is NP-complete \cite{yanna82}. Therefore, it is NP-complete to determine whether an $n$-element comparability graph has \textit{prn} $k$.  

The class of complete graphs is precisely the class of graphs having the representation number and the \textit{prn} equal to one. While the class of graphs with representation number at most two is characterized by the class of circle graphs in \cite{halldorsson11}, the class of permutation graphs precisely has the \textit{prn} at most two (cf. \cite{KK_KM_3}). Although examples of graphs with greater representation numbers and \textit{prn} exist, in both contexts, no characterization of graphs is available for other numbers.

A bipartite graph is a comparability graph that is precisely isomorphic to the cover graph of its induced poset. Determining the representation number and \textit{prn} of bipartite graphs are open problems. In \cite{yanna82}, it was shown that deciding whether a bipartite graph has \textit{prn} $k$, for fixed $k \ge 4$ is NP-complete. Further, in \cite{felsner2017complexity}, it was established that deciding whether a bipartite graph has \textit{prn} three is also NP-complete. Accordingly, in the literature, various authors have focused on determining these numbers or bounds on these numbers for specific subclasses of bipartite graphs. For example, these numbers are determined for even cycles, trees, and crown graphs (cf. \cite{glen18,halldorsson16,KK_KM_3}). The representation number of prisms, ladder graphs, outerplanar graphs, $G_k$ graphs, and complete bipartite graphs are known (cf. \cite{halldorsson11,kitaev13,kitaev15mono,kitaev08}). Upper bounds for the representation number of hypercubes, 3-colorable graphs, etc., were determined in \cite{broere_2019,halldorsson16}. Certain bounds on the \textit{prn} of bipartite graphs are given through the notion called boxicity in \cite{sunil11a}. However, determining boxicity for bipartite graphs is NP-hard \cite{sunil11b}. The upper bound on the \textit{prn} of bipartite graphs given in \cite{Broere18} is the best known in the literature. 

In this paper, we reconcile the relation between the \textit{prn} of a comparability graph and the dimension of its induced poset, and observe that they are the same. Accordingly, we present the classes of graphs for which the \textit{prn} are known (see Section \ref{prn=dim}). Then, in Section \ref{gen_prop}, we provide some general properties of the permutations utilized in representing a bipartite graph; these properties play a vital role in this work. Focusing on the \textit{prn} of bipartite graphs, we employ the notion of a neighborhood graph and give a polynomial-time procedure to construct permutations on the vertices of a given bipartite graph whose concatenation represents the bipartite graph. This gives us an upper bound for the \textit{prn} of bipartite graphs (cf. Section \ref{main}). In Section \ref{comp_broere}, we show that this is an improved upper bound in comparison to the best available one given by Broere in \cite{Broere18}. In Section \ref{k0_class}, we provide a criterion for determining the \textit{prn} of a special class and, further, give examples in Section \ref{kappa0_ex}. Moreover, in Section \ref{Example-ecg}, we construct a class of bipartite graphs, viz., extended crown graphs, defined over posets, and investigate their \textit{prn}. 
\newpage

\section{Preliminaries}

	This section provides the necessary background material on words, graphs, word-representable graphs and posets, and fix the notation. For the concepts that are not presented here, one may refer to  \cite{golumbic,kitaev15mono,Trotter92,west}.  
	
	Let $A$ be a finite set. A word over $A$ is a finite sequence of elements of $A$. The words are written by juxtaposing the symbols of the sequence. A subword $x$ of a word $y$ is a subsequence of the sequence $y$, denoted by $x \le y$. For instance, $abaa \le aabbaba$. Given a word $x$ over $A$ and a subset $B$ of $A$, the subword of $x$ consisting only of the symbols of $B$ is denoted by $x_B$. For example, if $x = abbcbaccb$, then $x_{\{a, b\}} = abbbab$.   We say the letters $a$ and $b$ alternate in a word $x$, if $x_{\{a, b\}}$ is of the form $abab\cdots$ or $baba\cdots$ (of even or odd length). A word $x$ is said to be $k$-uniform if the number of occurrences of each letter in $x$ is $k$. 
	
	In this paper, we only consider the graphs which are simple (i.e., without loops or parallel edges) and undirected. To distinguish between two-letter words and edges of graphs, we write $\overline{ab}$ to denote an undirected edge between vertices $a$ and $b$. A directed edge from $a$ to $b$ is denoted by $\overrightarrow{ab}$. The neighborhood of a vertex $a$ in a graph, denoted by $N(a)$, is the set of all vertices adjacent to $a$ in the graph. An orientation of a graph is an assignment of direction to each edge so that the resulting graph is a directed graph. An orientation of a graph is said to be a transitive orientation if the adjacency relation on the vertices in the resulting directed graph is transitive. A graph is said to be a comparability graph if it admits a transitive orientation.
	
	A graph $G$ is said to be a bipartite graph if its vertex set can be partitioned into $\{A, B\}$, called bipartition, such that every edge in $G$ connects a vertex in $A$ to a vertex in $B$. 
	A bipartite graph with bipartition $\{A, B\}$, where $|A| = m$ and $|B| = n$, is said to be complete, denoted by $K_{m,n}$, if every vertex in $A$ is adjacent to all vertices of $B$.
	
	A graph $G = (V, E)$ is said to be word-representable if there is a word $x$ over the set $V$ such that two vertices $a$ and $b$ are adjacent in $G$ if and only if $a$ and $b$ alternate in $x$. If a $k$-uniform word represents $G$, it is said to be $k$-word-representable. Every word-representable graph is $k$-word-representable, for some $k$ \cite{kitaev08}. The smallest number $k$ such that $G$ is $k$-word-representable is the representation number of a word-representable graph $G$, denoted by $\mathcal{R}(G)$. 
	
	A graph $G$ is said to be permutationally representable if it can be represented by a word of the form $p_1 p_2 \cdots p_k$, where each $p_i$ is a permutation of the vertices of $G$. Furthermore, if a graph can be represented permutationally with $k$ permutations, it is called a permutationally $k$-representable graph. The permutation-representation number (in short, \textit{prn}) of $G$, denoted by $\mathcal{R}^p(G)$, is the smallest number $k$ such that $G$ is permutationally $k$-representable. It is clear that $\mathcal{R}(G) \le \mathcal{R}^p(G)$. 
	
	For a complete bipartite graph $K_{m,n}$ with at least three vertices, we have $\mathcal{R}^p(K_{m,n}) = \mathcal{R}(K_{m,n}) = 2$.  It was shown in \cite[see the proof of Theorem 4]{halldorsson11} that $\mathcal{R}^p(H_{n,n}) = n$ (for $n \ge 2$) and in \cite{glen18} that $\mathcal{R}(H_{n,n}) = \lceil n/2 \rceil$  (for $n \ge 5$), where $H_{n,n}$ is the graph, called a crown graph, obtained by removing a perfect matching from $K_{n,n}$. Further, it was conjectured in \cite[see Conjecture 1]{glen18} that every bipartite graph on $n$ vertices has representation number at most $n/4$. 
	
	A set $P$ with a partial order, say $<$ on $P$, is called a partially ordered set or simply a poset. Two elements $a, b \in P$ are comparable if $a < b$ or $b < a$; otherwise, they are incomparable. A chain in $P$ is a set of pairwise comparable elements of $P$, whereas an antichain is a set of pairwise incomparable elements of $P$. A chain cover of a poset $P$ is a collection of disjoint chains whose union is $P$. The width of $P$, denoted by $w(P)$, is the size of a largest antichain in $P$. The height of $P$ is the size of a largest chain in $P$. It was established in \cite{Dilworth50} that for any finite poset, the size of a smallest chain cover is equal to the width of the poset, known as Dilworth's Theorem. A polynomial-time method for determining the width and a smallest chain cover of a poset  was provided in \cite{Felsner03,hopcroft1973n}. 
	
	A linear order of a set $P$ is a partial order on $P$ in which all elements of $P$ form a chain. A collection of linear orders $\{L_1,L_2, \ldots, L_k\}$ of a poset $P$ is called a realizer of $P$, if  $\cap^k_{i=1} L_i$ is the partial order on $P$,  i.e., for every $a , b \in P$, $a < b$ in $P$ if and only if $a < b $ in all $L_i$. The dimension of a poset $P$  is the smallest positive integer $k$ such that $P$ has a realizer of size $k$, denoted by $\dim(P)$ \cite{dushnik1941partially}.

	A comparability graph $G$ induces a poset denoted by $P_G$ based on the chosen transitive orientation. The cover graph of a poset $P$, denoted by $G_P$, is the graph whose vertex set is the set $P$, and two elements $a , b \in P$ are adjacent in $G_P$ if and only if $a < b$ in $P$ and there is no element $c$ in $P$ such that $a < c < b$. The Hasse diagram of a poset $P$ is a drawing of the cover graph $G_P$ in which $b$ is placed above $a$ whenever $a < b$ in $P$. In this paper, Hasse diagrams are used to depict the examples of posets.

	It is clear from the following result that the \textit{prn} of a disconnected graph is equal to the maximum of the respective numbers of its components.  
	
	\begin{lemma}[\cite{Broere18}] \label{disconnected}
		Let $G = (V,E)$ and $H = (V',E')$ be graphs with $V \cap V' = \emptyset$ that can be permutationally represented by $p_1p_2 \cdots p_k$ and $ q_1q_2 \cdots q_l$ respectively, where $2 \le k \le l$. Then, $G \cup H = (V \cup V', E\cup E')$ can be permutationally represented by $ p_1q_1p_2q_2 \cdots p_{k-1}q_{k-1}q_kp_kq_{k+1}p_k \cdots q_lp_k$.  
	\end{lemma}
	
	\begin{corollary}
		Suppose $\mathcal{R}^p(G)=k$ and $\mathcal{R}^p(H)=l$. Then, $\mathcal{R}^p(G \cup H)= \max \{k,l\}$.
	\end{corollary}
	
	Accordingly, it is sufficient to focus on determining the \textit{prn} of connected bipartite graphs. In this paper, unless specified otherwise, a bipartite graph with bipartition $\{A, B\}$, written  $G = (A \cup B, E)$, is a connected graph. 

\section{Comparability Graphs: Dimension vs \textit{prn}}
\label{prn=dim}

In this section, we reconcile the results that reveal the connection between the \textit{prn} of a comparability graph and the dimension of its induced poset, and observe that they are the same. Accordingly, we present the upper bounds for the \textit{prn} of comparability graphs, particularly for bipartite graphs. First, we recall the following result. 

\begin{lemma}[\cite{halldorsson11}] 
	If $G$ is a comparability graph, then $\mathcal{R}^p(G) \le k$ if and only if $\dim(P_G) \le k$.
\end{lemma}

In fact, it was shown in \cite{halldorsson11} that every realizer of $P_G$ is precisely a set of permutations of vertices of $G$ whose concatenation represents $G$ permutationally. Hence, considering such a minimal set on either side, we can deduce the following corollary.  

\begin{corollary}
	If $G$ is a comparability graph, then $\mathcal{R}^p(G) = k$ if and only if $\dim(P_G) = k$.
\end{corollary}

It was shown in \cite{trotter_dimension} that posets which have the same underlying comparability graph have the same dimension. The posets of dimension at most two can be recognized in linear time (see \cite{ross97}), and these posets were characterized in terms of their underlying comparability graphs, known as permutation graphs \cite{Baker}. The class of permutation graphs is precisely the intersection of comparability and co-comparability graphs\footnote{A co-comparability graph is a graph whose complement is a comparability graph.} \cite{dushnik1941partially}. Moreover, permutation graphs were also characterized in terms of forbidden induced subgraphs in \cite{Gallai} and forbidden Seidel minors in \cite{Vincent}. 

For $k \ge 3$, the problem of recognizing the posets of dimension $k$ is NP-complete. While there were partial results on characterizing these classes, there is no complete characterization in the literature. However, certain upper bounds on the dimension of a poset are known. For instance, it was shown in \cite{hiraguti_55} that  $\dim(P)  \le \lfloor \frac{|P|}{2} \rfloor$, for a poset $P$ with $|P| \ge 4$. In \cite{Dilworth50}, it was proved that $\dim(P)  \le w(P)$, for any poset $P$. In addition, considering the results in \cite{kimble_extermal,trotter_inequalities}, we state the following:
\begin{theorem}[\cite{Dilworth50,kimble_extermal,trotter_inequalities}]
	Let $P$ be a poset with $n$ elements. Then, $\dim(P) \le \min\{w(P), n-w(P)\}$.
\end{theorem}

In the context of bipartite graphs, we deduce the following corollary.

\begin{corollary} \label{known_poset}
	For a bipartite graph $G = (A \cup B, E)$, $\mathcal{R}^p(G) \le \min\{|A|,|B|\}$.
\end{corollary} 

An improved upper bound on the \textit{prn} of bipartite graphs was provided as per the following result.

\begin{theorem}[\cite{Broere18}]\label{known_result}
	For a bipartite graph $G = (A \cup B, E)$, we have \[\mathcal{R}^p(G) \le \min \{\alpha, \beta\},\] where
	$\alpha = |\{N(a)|a\in A\}|$ and $\beta = |\{N(b)|b\in B\}|$.
\end{theorem}

In the literature, better upper bounds for dimension were obtained for specific types of posets, whose cover graphs are planar. A poset is said to be planar (or outerplanar) if its Hasse diagram is a planar (or outerplanar\footnote{A planar graph is called outerplanar if all its vertices lie on the unbounded face of a planar drawing.}) graph. For example, in \cite{trotter1977dimension}, it was shown that if the cover graph of a poset $P$ is a tree, then $\dim(P) \le 3.$ Also, if a planar poset $P$ has a maximum element, then $\dim(P) \le 3$.  Generalizing these classes, in \cite{felsner2015dimension}, it was proved that if the cover graph of a poset $P$ is outerplanar, then $\dim(P) \le 4$. Further, if the underlying comparability graph of a poset $P$ is planar, then 
$\dim(P) \le 4$.

Specific to bipartite graphs, in \cite{felsner2010adjacency}, it was proved that $\mathcal{R}^p(G) \le 4$ if a bipartite graph $G$ is planar. Further, if $G$ is outerplanar, then $\mathcal{R}^p(G) \le 3$ \cite{felsner2015dimension}. Without a restriction on planarity, it was shown in \cite{chaplick2018grid} that if a bipartite graph is a grid intersection graph\footnote{A grid intersection graph is an intersection graph of horizontal and vertical segments in the plane.}, then $\mathcal{R}^p(G) \le 4$.

Among the subclasses of bipartite graphs, in addition to crown graphs, the \textit{prn} of generalized crown graphs was obtained in \cite{trotter_74}. A generalized crown graph $S_n^k$ is a bipartite graph obtained by removing $(k+1)$ perfect matchings from the complete bipartite graph $K_{n+k,n+k}$. For $n \ge 3$ and $k \ge 0$, $\mathcal{R}^p(S_n^k) =  \Bigl \lceil \frac{2(n+1)k}{(k+2)}\Bigr \rceil$.
		
\section{General Properties of Permutations}
\label{gen_prop}

In this section, we prove some general properties of permutations used in representing a bipartite graph. These properties will play a vital role in constructing the respective permutations. Let $G = (A \cup B, E)$ be a bipartite graph represented by $p_1 \cdots p_k$, where $p_i$ ($1 \le i \le k$) is a permutation of the vertices of $G$.  

First,  we show that a vertex adjacent to two vertices will not appear in between them in any permutation $p_i$.

\begin{lemma}\label{lemma_1}
	If $\overline{ab}, \overline{cb} \in E$, then $abc \not \le p_i$, for all $1\le i \le k$. 
\end{lemma}
\begin{proof}
	Suppose  $abc  \le p_t$ for some $t$. By the choice of $a$ and $c$, they are not adjacent in the bipartite graph $G$. Hence, as $ac \le p_t$, there exists a permutation $p_s$ such that $ca \le p_s$. Further, since $\overline{ab} \in E$ and $ab \le p_t$, we have $ab \le p_i$, for all $1 \le i \le k$. Consequently, $cb \le p_s$. This is not possible, as $\overline{cb} \in E$ and $bc \le p_t$. Hence, $abc \not\le p_i$, for all $1\le i \le k$. \qed
\end{proof}

As a consequence of Lemma \ref{lemma_1}, we observe that if any vertex adjacent to a vertex $a$ appears on the right side (or left side) of $a$ in a permutation $p_i$, then every other vertex adjacent to $a$ also appears on the right side (respectively, left side) of $a$ in $p_i$.   

\begin{lemma}\label{corollary_1}
	For $\overline{ab} \in E$, if  $ab \le p_i$, then $ab' \le p_i$, for all $b' \in N(a)$ and $1 \le i \le k$.
\end{lemma}
\begin{proof}
	For some $b' \in N(a)$, if $b'a \le p_i$, then $b'ab \le p_i$. But by Lemma \ref{lemma_1}, this is not possible. Hence, $ab' \le p_i$.   \qed
\end{proof}

We now prove a result on appearance of neighbors of two vertices, say $a$ and $c$, in a permutation, when they have common neighbors. If the neighbors of $a$ appear on the right side (or left side) of $a$ in a permutation $p_i$, then the neighbors of $c$ also appear on the right side (respectively, left side) of $c$ in any $p_i$.   

\begin{lemma}\label{lemma_2}
	Suppose $N(a) \cap N(c) \ne \varnothing$, for some vertices $a, c$ of $G$.  For $1 \le i \le k$ and $b \in N(a)$, if $ab \le p_i$, then $cb' \le p_i$, for all $b' \in N(c)$.
\end{lemma}
\begin{proof}
	For $b \in N(a)$, suppose $ab \le p_i$. Let $b'' \in N(a) \cap N(c)$. By Lemma \ref{corollary_1}, $ab'' \le p_i$.  If $b''c \le p_i$, then $ab''c \le p_i$. This is not possible by Lemma \ref{lemma_1}. Hence, $cb'' \le p_i$. Again, by Lemma \ref{corollary_1}, $cb' \le p_i$, for all $b' \in N(c)$. \qed
\end{proof}

\section{Upper Bound for \textit{prn}}
\label{main}

In this main section, we present a polynomial-time procedure to obtain permutations of the vertices of a bipartite graph $G$ and provide an upper bound for the \textit{prn} of $G$.    

A graph is said to be a reduced graph if no two vertices have the same neighborhood. Every graph can be reduced by considering only the vertices with distinct neighborhoods and removing the remaining vertices. In \cite{Broere18}, it was shown that the \textit{prn} of a bipartite graph and the \textit{prn} of its reduced graph are the same by constructing their respective words. Accordingly, in what follows, we consider reduced bipartite graphs and discuss their \textit{prn}.

\subsection{Polynomial-time Procedure}
\label{poly_proc}

Let ${S} = \{a_{i_1}, a_{i_2}, \ldots, a_{i_k}\}$ be a set with indices $i_1 < i_2 < \cdots < i_k$. We write $\underline{S}$ to represent the word $a_{i_k} a_{i_{k - 1}} \cdots a_{i_1}$ in which the symbols of $S$ appear exactly once and are arranged in such a way that the indices are in decreasing order. Similarly, we write  $\overline{S}$ to represent the word $a_{i_1} a_{i_2} \cdots a_{i_k}$ in which the symbols of $S$ appear exactly once and are arranged such that the indices are in increasing order.

Let $G=(V, E)$ be a graph and $A \subseteq V$. We define a neighborhood graph of $G$ with respect to $A$ as the directed graph $\mathcal{N}_A(G) = (A, E')$, where $\overrightarrow{ab} \in E'$ if and only if $N(a) \subseteq N(b)$ for all $a, b \in A$. If $A = V$, we may simply call it as a neighborhood graph of $G$, and it is denoted by $\mathcal{N}(G)$.

\begin{remark}
	Since the set containment relation is transitive, the graph $\mathcal{N}_A(G)$ is a comparability graph, and hence, it induces a poset. In the rest of the paper, the poset $P_{\mathcal{N}_A(G)}$ induced by $\mathcal{N}_A(G)$ is simply denoted by $P_A$.
\end{remark}

Let $G = (A \cup B, E)$ be a bipartite graph. Consider a chain cover of $P_A$ as per the following:
\begin{center}
	$X_1: a_{11}<a_{12}<a_{13}<\cdots<a_{1n_1}$\\
	$X_2: a_{21}<a_{22}<a_{23}<\cdots<a_{2n_2}$\\
	$\vdots$\\
	$X_k: a_{k1}<a_{k2}<a_{k3}<\cdots<a_{kn_k}$\\
\end{center} 
Corresponding to each chain $X_i$, we create a permutation $p_i$ of vertices of $A \cup C$, where $C$ is a relabeled set of $B$ as described below. 

For a set of vertices $U$ and a vertex $a$, let $U_{-a} = U \setminus N(a)$, the set of non-adjacent vertices of $a$ in the set $U$. 

Consider one of the chains, say $X_1$, and construct a permutation, say $p_1'$, as per the following. In $X_1$, since $N(a_{1(i-1)}) \subsetneq N(a_{1i})$, we write the exclusive neighbors of $a_{1i}$, i.e., $N(a_{1i})_{-a_{1(i-1)}}$, on the right side of $a_{1i}$ and on the left side of $a_{1(i-1)}$. Then, we write the vertices in $B$ which are not adjacent to any of the vertices in $X_1$, i.e., $B_{-a_{1n_1}}$, on the left side of $a_{1n_1}$. Now we write the remaining vertices of $G$, i.e., the vertices of the chains $X_2, \ldots, X_k$, on the left side of $B_{-a_{1n_1}}$. Although the vertices of $N(a_{1i})_{-a_{1(i-1)}}$ and $B_{-a_{1n_1}}$ can be written in any order, we will write them in the increasing order as per their indices. Thus,
\[p_1' = x_1 \overline{B_{-a_{1n_1}}} a_{1n_1} \overline{N(a_{1n_1})_{-a_{1(n_1-1)}}} a_{1(n_1-1)} \cdots a_{12} \overline{N(a_{12})_{-a_{11}}} a_{11}\overline{N(a_{11})},\] where $x_1 = X_2^{^\text{o}}\cdots X_k^{^\text{o}}$. Here, for $1 \le j \le k$, $X_j^{^\text{o}} = a_{j1}a_{j2}\cdots a_{jn_j}$, the word obtained by writing the vertices of $X_j$ in the increasing order as per the ordering of $P_A$. 

Consider $p'_{1_B}$ the subword of $p_1'$ with the vertices of $B$. Relabel the word $p'_{1_B}$ as $c_1c_2\cdots c_m$ (where $m = |B|$) with indices increasing from left to right. Thus, $C = \{c_1, \ldots, c_m\}$ is the relabeled set of $B$. Now let $p_1$ be the permutation obtained from $p_1'$ by replacing the relabeled vertices from $C$ in place of corresponding vertices of $B$. 

For $i = 2, \ldots, k$, based on the chain $X_i$, we construct the permutation $p_i$ with vertices of $A \cup C$ by 
\[p_i = x_i \underline{C_{-a_{in_i}}} a_{in_i} \underline{N(a_{in_i})_{-a_{i(n_i-1)}}} a_{i(n_i-1)} \cdots a_{i2} \underline{N(a_{i2})_{-a_{i1}}} a_{i1}\underline{N(a_{i1})},\] where $x_i = X_1^{^\text{o}}\cdots X_{i-1}^{^\text{o}}X_{i+1}^{^\text{o}}\cdots X_k^{^\text{o}}$ and the neighborhoods are considered with the vertices of $C$. 

Finally, set $p_0 =  X_k^{^\text{o}}X_{k-1}^{^\text{o}}\cdots  X_1^{^\text{o}}c_mc_{m-1}\cdots c_1$.

A demonstration of this method for obtaining a word that represents a bipartite graph permutationally is given in Appendix \ref{demo}.

\subsubsection{Time Complexity}

 First, recall that obtaining chain cover of a poset can be done in $O(n^3)$ time \cite{Felsner03}. Further, note that the construction of each permutation $p_i$ ($1 \le i \le k$) involves finding the neighbors of each vertex in $X_i$ and arranging them in order. Assume the graph is maintained in the adjacency matrix. Note that relabeling of vertices of $B$ takes $O(n)$ time. Except the labels, since $B$ and $C$ are the same, we will use the set $B$ for measuring the complexity in the following.  Also, maintain an array of vertices of $B$ in the increasing order of indices and another array of vertices of $B$ in the decreasing order of indices. Note that we can prepare these arrays in $O(n^2)$ time. Every time, while constructing the permutation $p_i$ (from its right side), take a copy of arranged vertices of $B$ in $O(n)$ time, say $B'$. Find the neighbors of $a_{i1}$ in $O(n)$ time from $B'$ and arrange them in the same order for constructing $p_i$ while deleting them in $B'$. Then, repeat the same for each $a_{ij}$ ($1 < j \le n_i$) to construct the permutation $p_i$ in $O(n^2)$ time, which includes arranging the prefix $x_i$. The permutation $p_0$ can be easily constructed in $O(n^2)$ time. As the number of permutations is $O(n)$, the entire process of constructing the permutation will take $O(n^3)$ time. 
		
\subsection{Improved Upper Bound}

In continuation, we now prove the correctness of the procedure given in Section \ref{poly_proc}, and subsequently establish an upper bound for $\mathcal{R}^p(G)$. Further, we mention a case for which the \textit{prn} attains the upper bound.   

	\begin{theorem}\label{rep_G'}
		The word $p_0p_1 \cdots p_k$ permutationally represents the bipartite graph $G'= (A \cup C, E')$,  where $$E'=\{\overline{ac}\; |\; \overline{ab} \in E  \ \text{(for $a \in A$, $b \in B$), and $b$ is relabeled as $c \in C$}\}.$$ 
	\end{theorem}
	\begin{proof}
		We show that two vertices $a$ and $c$ are adjacent in $G'$ if and only if $a$ and $c$ alternate in the word $p_0p_1 \cdots p_k$.
		
		Suppose $\overline{ac} \in E'$. Let $a \in X_s$, say $a = a_{st}$ for some $1 \le t \le n_s$. Since $c$ is in the neighborhood of $a$, $c$ appears on the right side of $a$ in $p_s$, i.e., $ac \le p_s$. Further, for all $i \ne s$, as $a \in X_s$, we have $a \le x_i$ so that $ac \le p_i$. Clearly, $ac \le p_0$.   Hence, $ac \le p_i$ for all $0 \le i \le k$ so that $a$ and $c$ alternate in the 	word $p_0p_1 \cdots p_k$.
		
		Conversely, suppose $a$ and $c$ are not adjacent in $G'$. We deal with this in three cases. In each case, we identify two permutations: one with the subword $ac$ and the other with the subword $ca$.
		\begin{itemize}
			\item Case 1: $a, c \in A$.
			\begin{itemize}
				\item Subcase 1.1: Suppose $a, c \in X_s$. Then, either $a < c$ or $c < a$. Without loss of generality, assume $a < c$. Then, $ca \le p_s$. Clearly, $ac \le p_0$.  
				\item Subcase 1.2: Suppose $a \in X_s$ and $c \in X_t$, for $t \ne s$. Clearly, $ac \le p_t$ and $ca \le p_s$.
			\end{itemize}
			\item Case 2: $a \in A$ and $c \in C$. Clearly, $ac \le p_0$. Suppose $a \in X_s$. Then, note that $ca \le p_s$.
			\item Case 3: $a, c \in C$. Let $a = c_p$ and $c = c_q$. Without loss of generality, assume $p < q$. Then, $ac = c_pc_q \le p_1$ and $ca = c_qc_p \le p_0$.
		\end{itemize}
		Hence, $a$ and $c$ do not alternate in the word $p_0p_1 \cdots p_k$. \qed
	\end{proof}
	
	For $0 \le i \le k$, let $p_i'$ be the permutation of vertices of $A \cup B$ obtained from $p_i$ by relabeling the vertices of $C$ with their original labels of $B$. Then, clearly, the word $p_0'p_1' \cdots p_k'$ represents the graph $G$ as stated in the following corollary of Theorem \ref{rep_G'}. 
	
	\begin{corollary}
		The word $p_0'p_1' \cdots p_k'$ permutationally represents the graph $G$.
	\end{corollary}
	
	Further, since the width of $P_A$ is the size of a minimum chain cover (Dilworth's Theorem), we have the following corollary of Theorem \ref{rep_G'}. 
	
	\begin{corollary}
		$\mathcal{R}^p(G) \le 1 + w(P_A)$.
	\end{corollary}
	
	In view of the symmetry between $A$ and $B$, we also have $\mathcal{R}^p(G) \le 1 + w(P_B)$. Since the words representing $G$ can be obtained using chain covers with respect to $A$ or $B$, we have the following result:
	
	\begin{theorem}\label{upperbd}
		Let $G = (A \cup B, E)$ be a bipartite graph and $\kappa_0 = \min \{w(P_A), w(P_B)\}$. Then,  \[\mathcal{R}^p(G) \le 1 + \kappa_0.\]   
	\end{theorem}
	
	Every connected bipartite graph $G$ on at least five vertices has a crown graph as an induced subgraph. We call it as an induced crown graph of $G$. Based on the \textit{prn} for crown graphs, we state a lower bound on the \textit{prn} of bipartite graphs. Since the class of bipartite graphs is hereditary\footnote{A class of graphs is hereditary if it contains all induced subgraphs of its members.}, if $H$ is an induced subgraph of $G$, then $\mathcal{R}^p(H) \le \mathcal{R}^p(G)$. Therefore, we have the following result.
	
	\begin{lemma}\label{lowerbd}
		Let $G$ be a bipartite graph and $H_{k,k}$ be a largest induced crown graph of $G$. Then,	$\mathcal{R}^p(G) \ge k$.
	\end{lemma}

	From the above results, we can conclude the \textit{prn} of a particular type of bipartite graph.
	
	\begin{theorem} \label{main_result}
		Let $G = (A\cup B, E)$ be a bipartite graph and $H_{k,k}$ be a largest induced crown graph of $G$. If $k = \kappa_0$, then $\mathcal{R}^p(G) = \kappa_0 \ \text{or} \ 1 + \kappa_0$.
	\end{theorem}

\section{Class with \textit{prn} at most $\kappa_0$}
\label{k0_class}

In this section, we provide a criterion for the chain covers of $P_A$ (or $P_B$) such that the given bipartite graph $G = (A \cup B, E)$ is permutationally $\kappa_0$-representable. 

\begin{notation}
	Let $X$ be a chain in $P_A$ and $b\in B$. The relative neighborhood of $b$ with respect to $X$, i.e. $N(b) \cap X$, is denoted by $N_X(b)$.
\end{notation}

\begin{lemma}
	Suppose $X$ is a chain in $P_A$. For all $b,b'\in B$, the relative neighborhoods $N_X(b)$ and $N_X(b')$ are comparable, i.e., $N_X(b)\subseteq N_X(b')$ or $N_X(b')\subseteq N_X(b)$.
\end{lemma}
\begin{proof}
	Let $X:a_1<a_2<\cdots<a_n$. Assume that $N_X(b)$ and $N_X(b')$ are incomparable, i.e.,  $N_X(b)\not \subseteq N_X(b')$ and $N_X(b')\not \subseteq N_X(b)$. Then, $N_X(b) \not = \varnothing$ and $N_X(b')\not = \varnothing$. There exists $a_r \in N_X(b)\backslash N_X(b')$ and $a_s \in N_X(b')\backslash N_X(b)$. As $a_r, a_s \in X$, without loss of generality, assume $a_r<a_s$. Then, $N(a_r)\subseteq N(a_s)$. This implies that $b \in N(a_s)$, which is a contradiction. Hence,  $N_X(b)$ and $N_X(b')$ are comparable. \qed
\end{proof}

If the chain cover has only one chain, it is evident that $\mathcal{R}^p(G) = 2$ (see Example \ref{ex_pa_chain}). Assume $M=\{X_1,X_2,\ldots, X_k\}$ ($k \ge 2$) is a chain cover of poset $P_A$. Let $X_t \in M$. For some $b, b'\in B$, if $N_{X_t}(b) = N_{X_t}(b') = \{a_{tl},a_{t(l+1)},\dots, a_{tn_t}\}$, where $1 \le l\le n_t$, then the permutation corresponding to $X_t$ is of the form
\[p_t : x_t \underline{B_{-a_{tn_t}}} a_{tn_t} \underline{N(a_{tn_t})_{-a_{t(n_t-1)}}} a_{t(n_t-1)} \cdots a_{tl} \underline{N(a_{tl})_{-a_{t(l-1)}}} a_{t(l-1)}\cdots   a_{t1}\underline{N(a_{t1})},\] where $b, b'\in N(a_{tl})_{-a_{t(l-1)}}$. 

Based on this observation, we make a minor modification to the method described in Section \ref{poly_proc}, i.e., for all $b,b'\in B$ if $N_{X_1}(b)=N_{X_1}(b')$ and $N(b')\subsetneq N(b)$, then put $b$ before $b'$ in $p_1$ such that $b=c_p$ and $b'=c_q$, where $p<q$.

\begin{remark}\label{nbd_compare}
	Note that the neighborhoods of no two vertices in $G$ are the same. As $X_i \cap X_j = \varnothing$ for $i \ne j$, we have $N(b') \subsetneq N(b)$ if and only if $N_{X_i}(b') \subseteq N_{X_i}(b)$ for all chains $X_i$ and $N_{X_t}(b') \subsetneq N_{X_t}(b)$ for some chain $X_t$.
\end{remark}

\begin{theorem}\label{k-rep}
	Suppose $M=\{X_1,X_2,\ldots, X_k\}$ $(k \ge 2)$ is a chain cover of the poset $P_A$. If $M$ satisfies the following condition with respect to $B$, then $G$ is permutationally $k$-representable.  
	
	\begin{tabular}{lc}
		For all $b, b' \in B$, if there exist $i, j$ with $i \ne j$ \\ such that $N_{X_i}(b)  \subseteq N_{X_i}(b')$ and $N_{X_j}(b')  \subseteq N_{X_j}(b)$. & \qquad \qquad$\cdots\cdots(\zeta_B)$
	\end{tabular}
\end{theorem}

\begin{proof}
	Let $G'$ be the relabeled graph of $G$ and $w'=p_1p_2\cdots p_k$ be the word where $p_1,p_2,\ldots, p_k$ are permutation of vertices of $G'$ as per the construction given in Section \ref{poly_proc}. It is sufficient to show that $w'$ represents $G'$ permutationally as the vertices of $G'$ can be relabeled back to their original labeling in $w'$, which permutationally represents $G$.
	
	We show that two vertices $a$ and $c$ are adjacent in $G'$ if and only if $a$ and $c$ alternate in the word $w'$.
	
	Suppose $\overline{ac} \in E'$. Let $a \in X_t$, say $a = a_{tl}$ for some $1 \le l \le n_t$. Since $c$ is in the neighborhood of $a$, $c$ appears on the right side of $a$ in $p_t$, i.e., $ac \le p_t$. Further, for all $i \ne t$, as $a \in X_t$, we have $a \le x_i$ so that $ac \le p_i$. Hence, $ac \le p_i$ for all $1 \le i \le k$ so that $a$ and $c$ alternate in the word $w'$.
	
	Conversely, suppose $a$ and $c$ are not adjacent in $G'$. We deal with this in three cases. In each case, we identify two permutations: one with the subword $ac$ and the other with the subword $ca$.
	\begin{itemize}
		\item Case 1: $a, c \in A$.
		\begin{itemize}
			\item Subcase 1.1: Suppose $a, c \in X_s$. Then, either $a < c$ or $c < a$. Without loss of generality, assume $a < c$. Then, $ca \le p_s$. Since $a, c \in X_s$, there exists $X_t$ ($t \ne s$) such that $a, c \not\in X_t$. Then, $ac \le X_s^{^\text{o}} \le x_t \le p_t$.  
			\item Subcase 1.2: Suppose $a \in X_s$ and $c \in X_t$, for $t \ne s$. Then, $ac \le X_s^{^\text{o}}c \le x_tc \le p_t$. Similarly, $ca \le p_s$.
		\end{itemize}
		\item Case 2: $a \in A$ and $c \in C$. Suppose $a \in X_s$. As $c \notin N(a)$, clearly $ca \le p_s$. Since $a \in X_s$, there exists $X_t$ ($t \ne s$) such that $a \not\in X_t$. Then,  $ac \le X_s^{^\text{o}}c \le x_tc \le p_t$. 
		\item Case 3: $a, c \in C$. Let $a = c_p$ and $c = c_q$. Without loss of generality, assume $p < q$. Then, $ac = c_pc_q \le p_1$, which means that $N_{X_1}(a) \subseteq N_{X_1}(c)$. As $N_{X_i}(a)\subseteq N_{X_i}(c)$ and $N_{X_j}(c)\subseteq N_{X_j}(a)$ for some $i$ and $j$, there exist a chain $X_s$ such that $N_{X_s}(c)\subseteq N_{X_s}(a)$.
		\begin{itemize}
			\item  Suppose $N_{X_s}(c) =  N_{X_s}(a)$. As $c_p$ and $c_q$ are in the same exclusive neighborhood in $p_s$ and are written in decreasing order, we have $c_qc_p = ca \le p_s$.
			\item  Suppose $N_{X_s}(c) \subsetneq  N_{X_s}(a)$. Then, there exists $a_r\in N_{X_s}(a) \setminus N_{X_s}(c)$ such that $ca 	\le c a_r a \le p_s$.                                                                                                                                                                                                                                                                                                                                                                                                                                                                                                                                                                                                                                                                                                                                                                                                                                                                                                                                                                                                                                   
		\end{itemize} 		
	\end{itemize}
	Hence, $a$ and $c$ do not alternate in $w'$. \qed	
\end{proof}

\begin{corollary}\label{kappa0_rep}
	If $\min\{w(P_A), w(P_B)\} = w(P_A)$ (or $= w(P_B)$), a minimal chain cover of $P_A$ (or $P_B$) satisfies the condition $\zeta_B$ (or $\zeta_A$), then $G$ is permutationally $\kappa_0$-representable.
\end{corollary}

\begin{remark}\label{anti_k0-rep}
	Let $G= (A \cup B, E)$ be a bipartite graph such that both $P_A$ and $P_B$ are antichains. Then, $P_A$ (also $P_B$) has only one chain cover. For all $u, v \in A$ (also in $B$), since $N(u)$ and $N(v)$ are incomparable, the chain cover of $P_A$ (of $P_B$) satisfies $\zeta_B$ (respectively, $\zeta_A$). Hence, $\mathcal{R}^p(G) \le \kappa_0$. In particular, every regular bipartite graph, in which all vertices have the same degree, is permutationally $\kappa_0$-representable. 
\end{remark}

\section{Comparison with Broere's Result}
\label{comp_broere}

We now show that the upper bound for \textit{prn} of a bipartite graph obtained in Theorem \ref{upperbd} is an improvement of the upper bound given by Broere (Theorem \ref{known_result}). 

Let $G = (A \cup B, E)$ be a reduced bipartite graph. Accordingly, the number of distinct neighborhoods of $A$ is $|A|$, i.e., $\alpha = |\{N(a)|a\in A\}| = |A|$. Also, the number of distinct neighborhoods of $B$ is $|B|$, i.e., $\beta = |\{N(b)|b \in B\}| = |B|$. Hence, since $w(P_A) \le |A|$ and $w(P_B) \le |B|$, we have $\kappa_0 \le \min\{\alpha, \beta\}$. 

In case $\mathcal{R}^p(G) = 1 + \kappa_0$, we observe that $\kappa_0 < \min\{\alpha, \beta\}$. On the contrary, suppose $\kappa_0 = \min\{\alpha, \beta\}$. We arrive at a contradiction as shown below in two cases. Without loss of generality, assume $\alpha \le \beta$. Then, $\kappa_0 = \alpha = |A|$. 
\begin{itemize}
	\item If $w(P_A) \le w(P_B)$, then $w(P_A) = \kappa_0 = |A|$ so that $P_A$ is an antichain. Hence, by Remark \ref{anti_k0-rep}, $G$ is permutationally $\kappa_0$-representable. But $\mathcal{R}^p(G) = 1 + \kappa_0$.
	
	\item Else, $w(P_B) = \kappa_0 < w(P_A) \le |A| = \min\{\alpha, \beta\}$. 
\end{itemize}

Hence, if $\mathcal{R}^p(G) \le \kappa_0$, then clearly $\kappa_0 \le \min\{\alpha, \beta\}$. And if $\mathcal{R}^p(G) = 1 + \kappa_0$, we have $\kappa_0 < \min\{\alpha, \beta\}$. 

\begin{example}
\begin{figure}[!h]
	\centering
	\begin{minipage}{.6\textwidth}
		\centering
		\[\begin{tikzpicture}[scale=0.5]			
			\vertex (1) at (0.5,0) [label=below:$1$] {};  
			\vertex (2) at (1.5,0) [label=below:$2$] {};
			\vertex (3) at (2.5,0) [label=below:$3$] {};
			\vertex (4) at (3.5,0) [label=below:$4$] {};
			\vertex (5) at (0,2) [label=above:$5$] {};  
			\vertex (6) at (1,2) [label=above:$6$] {};
			\vertex (7) at (2,2) [label=above:$7$] {};
			\vertex (8) at (3,2) [label=above:$8$] {};
			\vertex (9) at (4,2) [label=above:$9$] {}; 		
			\path
			(1) edge (5)
			(1) edge (7)
			(2) edge (5)
			(2) edge (6)
			(2) edge (8)
			(3) edge (9)
			(3) edge (7)
			(4) edge (7)
			(4) edge (9);
		\end{tikzpicture}\]
		\caption{A bipartite graph}
		\label{fig_G}
	\end{minipage}
\end{figure}
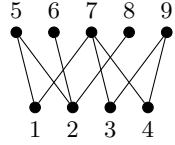
Consider the bipartite graph $G$ given in Fig. \ref{fig_G}. In \cite{Broere18}, Broere observed that $G$ is permutationally 3-representable. However, note that the reduced graph of $G$ satisfies Corollary \ref{kappa0_rep} and $\kappa_0 = 2$. Hence, $\mathcal{R}^p(G) = 2$. 
\end{example}

\section{Examples}
\label{kappa0_ex}

In this section, we consider some examples of bipartite graphs and determine their \textit{prn} using the method given in this work. Further, we also give their representation number.


\begin{example}\label{ex_pa_chain}
Let $G = (A\cup B, E)$ be a bipartite graph with $A = \{a_1,a_2,\ldots, a_n\}$. If $N(a_1) \subseteq N(a_2) \subseteq \cdots \subseteq N(a_n) = B$, then the poset $P_A$ is a chain. Note that $P_B$ is also a chain. Hence, $\mathcal{R}^p(G) = \mathcal{R}(G) = 2$.
\end{example}

\begin{example}		
\label{deg1}
		Let $G = (A\cup B, E)$ be a bipartite graph containing $H_{k, k}$ as its largest induced crown graph, and all other vertices have degree one. As the graphs under consideration are reduced, all the vertices of degree one are adjacent to distinct vertices. Note that for the corresponding reduced graph,  we will extend the graph $G$ to a graph $H = (C \cup D, F)$ such that $|C| = |D| = 2k$ by adding vertices, if required, of degree one in both the partite sets. Note that $G$ is an induced subgraph of $H$. The poset $P_C$ induced by the neighborhood graph of $H$ with respect to $C$ (as well as $P_D$) will have a depiction of a crown graph. An example of $G$ for $k = 5$ and its extension $H$ along with $P_C$ are depicted in Fig. \ref{crowns}.  Using Corollary \ref{kappa0_rep}, $H$ is permutationally $k$-representable. Further, since $H_{k, k}$ is an induced subgraph of $H$ as well as in $G$, using Theorem \ref{main_result}, we have $\mathcal{R}^p(G) = \mathcal{R}^p(H) = k$. 
		
		Note that the graphs $G$ and $H$ are obtained by successive addition of vertices with a new edge to the crown graph $H_{k, k}$. Hence, by \cite[Proposition 15]{kitaev13}, $\mathcal{R}(G) = \mathcal{R}(H) = \lceil \frac{k}{2} \rceil$, where $k \ge 5$. For $k = 4$, 	$\mathcal{R}(G) = \mathcal{R}(H) = 3 $, and $\mathcal{R}(G) = \mathcal{R}(H) = 2 $, for $k \le 3$.
			\begin{figure}[tbh!]
			\centering
				\begin{minipage}{.3\textwidth}
				\centering
				\begin{tikzpicture}[scale=0.6]
				\vertex (1) at (1,0) [label=left:$ 1$] {};  		
				\vertex (2) at (2,0) [label=left:$2$] {};
				\vertex (3) at (3,0) [label=left:$3$] {};		
				\vertex (4) at (4,0) [label=right:$4 $] {};			
				\vertex (6) at (1,-3) [label=below:$6$] {};
				\vertex (7) at (2,-3) [label=below:$7$] {};
							
				\vertex (11) at (1,-1.5) [label=left:$ 1'$] {};  		
				\vertex (12) at (2,-1.5) [label=left:$2'$] {};
				\vertex (13) at (3,-1.5) [label=left:$3'$] {};		
				\vertex (14) at (4,-1.5) [label=right:$4'$] {};			
				\vertex (18) at (3,1.5) [label=above:$8'$] {};
				\vertex (19) at (2,1.5) [label=above:$7'$] {};
				\path					
				(2) edge (11)
				(3) edge (11)
				(4) edge (11)
				(1) edge (12)
				(1) edge (13)
				(1) edge (14)
				(2) edge (13)
				(2) edge (14)
				(3) edge (12)
				(3) edge (14)
				(4) edge (12)
				(4) edge (13)
				(3) edge (18)
				(2) edge (19)
				(11) edge (6)
				(12) edge (7);
				\end{tikzpicture}
				
				$G$
			\end{minipage}
			\begin{minipage}{.3\textwidth}
				\centering
				\begin{tikzpicture}[scale=0.6]
				\vertex (1) at (1,0) [label=left:$ 1$] {};  		
				\vertex (2) at (2,0) [label=left:$2$] {};
				\vertex (3) at (3,0) [label=left:$3$] {};		
				\vertex (4) at (4,0) [label=right:$4 $] {};			
				\vertex (6) at (1,-3) [label=below:$6$] {};
				\vertex (7) at (2,-3) [label=below:$7$] {};
				\vertex (8) at (3,-3) [label=below:$8$] {};
				\vertex (9) at (4,-3) [label=below:$9$] {};
				\vertex (11) at (1,-1.5) [label=left:$ 1'$] {};  		
				\vertex (12) at (2,-1.5) [label=left:$2'$] {};
				\vertex (13) at (3,-1.5) [label=left:$3'$] {};		
				\vertex (14) at (4,-1.5) [label=right:$4'$] {};			
				\vertex (16) at (1,1.5) [label=above:$6'$] {};
				\vertex (17) at (2,1.5) [label=above:$7'$] {};
				\vertex (18) at (3,1.5) [label=above:$8'$] {};
				\vertex (19) at (4,1.5) [label=above:$9'$] {};
				\path	
				(2) edge (11)
				(3) edge (11)
				(4) edge (11)
				(1) edge (12)
				(1) edge (13)
				(1) edge (14)
				(2) edge (13)
				(2) edge (14)
				(3) edge (12)
				(3) edge (14)
				(4) edge (12)
				(4) edge (13)
				(1) edge (16)
				(2) edge (17)
				(3) edge (18)
				(4) edge (19)
				(11) edge (6)
				(12) edge (7)
				(13) edge (8)
				(14) edge (9);
			\end{tikzpicture}
			
			$H$
			\end{minipage}
			\begin{minipage}{.3\textwidth}
				\centering
				\begin{tikzpicture}[scale=0.6]
				\vertex (1) at (1,0) [label=above:$ 1$] {};  		
			\vertex (2) at (2,0) [label=above:$2$] {};
			\vertex (3) at (3,0) [label=above:$3$] {};		
			\vertex (4) at (4,0) [label=above:$4 $] {};			
			\vertex (6) at (1,-1.5) [label=below:$6$] {};
			\vertex (7) at (2,-1.5) [label=below:$7$] {};
			\vertex (8) at (3,-1.5) [label=below:$8$] {};
			\vertex (9) at (4,-1.5) [label=below:$9$] {};
			\path
			(6) edge (2)
			(6) edge (3)
			(6) edge (4)
			(7) edge (1)
			(7) edge (3)
			(7) edge (4)
			(8) edge (1)
			(8) edge (2)
			(8) edge (4)
			(9) edge (1)
			(9) edge (3)
			(9) edge (2);
			\end{tikzpicture}
			
			$P_C$
			\end{minipage}
			\caption{Depiction for Example \ref{deg1}}
			\label{crowns}
		\end{figure}
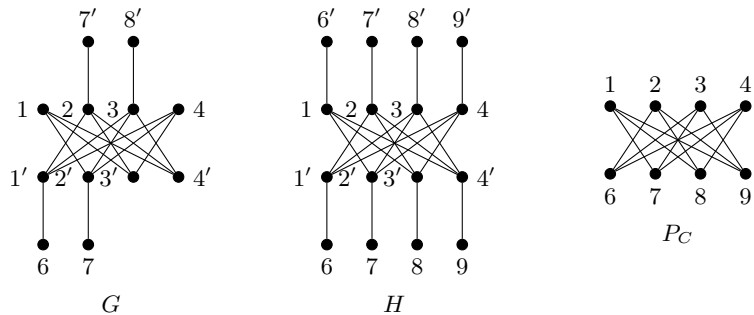
\end{example}

\begin{example}		
\label{all_adj}
			Let $G^k = (A\cup B, E)$ be a bipartite graph containing $H_{k, k}$ as its largest induced crown graph, and every other vertex is adjacent to all vertices of the other partite set. It is sufficient to consider the reduced version of $G^k$ in which at most one vertex in $A$ and at most one in $B$ have neighborhoods $B$ and $A$, respectively. An example of $G^5$ and its $P_A$ are depicted in Fig. \ref{star}. Note that the depiction of $P_A$ (and $P_B$) is a star. Then, the neighborhood graph with respect to $A$ and $B$ is a star. Using Corollary \ref{kappa0_rep} and Theorem \ref{main_result}, we have $\mathcal{R}^p(G^k) = k$, for $k \ge 2$. It can observed that $\mathcal{R}^p(G^1) = 2$.  
		
		\begin{figure}[h!]
			\centering
			\begin{minipage}{.6\textwidth}
				\begin{tabular}{ccc}
				\begin{tikzpicture}[scale=0.6]
				\vertex (1) at (1,0) [label=above:$ 1$] {};  		
				\vertex (2) at (2,0) [label=above:$2$] {};
				\vertex (3) at (3,0) [label=above:$3$] {};		
				\vertex (4) at (4,0) [label=above:$4 $] {};			
				\vertex (5) at (5,0) [label=above:$5$] {};
				\vertex (6) at (6,0) [label=above:$6$] {};

				\vertex (11) at (1,-1.5) [label=below:$ 1'$] {};  		
				\vertex (12) at (2,-1.5) [label=below:$2'$] {};
				\vertex (13) at (3,-1.5) [label=below:$3'$] {};		
				\vertex (14) at (4,-1.5) [label=below:$4'$] {};			
				\vertex (15) at (5,-1.5) [label=below:$5'$] {};
				\vertex (16) at (6,-1.5) [label=below:$6'$] {};
			
				\path	
				
				(2) edge (11)
				(3) edge (11)
				(4) edge (11)
				(5) edge (11)
				(1) edge (12)
				(1) edge (13)
				(1) edge (14)
				(1) edge (15)
				(2) edge (13)
				(2) edge (14)
				(2) edge (15)
				(3) edge (12)
				(3) edge (14)
				(3) edge (15)
				(4) edge (12)
				(4) edge (13)
				(4) edge (15)
				(5) edge (12)
				(5) edge (13)
				(5) edge (14) 
				(6) edge (11)
				(6) edge (12)
				(6) edge (13)
				(6) edge (14)
				(6) edge (15)
				(6) edge (16)
				(16) edge (1)
				(16) edge (2)
				(16)  edge (3)
				(16) edge (4)
				(16) edge (5);
				\end{tikzpicture}&\qquad\qquad&
				\begin{tikzpicture}[scale=0.6]
				\vertex (1) at (1,0) [label=below:$ 1$] {};  		
				\vertex (2) at (2,0) [label=below:$2$] {};
				\vertex (3) at (3,0) [label=below:$3$] {};		
				\vertex (4) at (4,0) [label=below:$4 $] {};			
				\vertex (5) at (5,0) [label=below:$5$] {};
				\vertex (6) at (3,1.5) [label=above:$6$] {};
				
				\path
				(1) edge (6)
				(2) edge (6)
				(3) edge (6)
				(4) edge (6)
				(5) edge (6);
				\end{tikzpicture}
			\end{tabular}
		\caption{Depiction for Example \ref{all_adj}}
		\label{star}
		\end{minipage}
		\begin{minipage}{.3\textwidth}
			\centering
			\begin{tikzpicture}[scale=0.4]				
				\vertex (1) at (0,0) [label=above:$ $] {};  
				\vertex (2) at (2,0) [label=above:$ $] {};
				\vertex (3) at (3,-1.5) [label=right:$ $] {};
				\vertex (4) at (2,-3) [label=below:$ $] {};
				\vertex (5) at (0,-3) [label=below:$ $] {};  
				\vertex (6) at (-1,-1.5) [label=left:$ $] {};
				\vertex (7) at (1,-1.5) [label=below:$ $] {};
				\path
				(1) edge(2)
				(2) edge (3)
				(3) edge (4)
				(4) edge (5)
				(5) edge (6)
				(6) edge (1)
				(7) edge (1)
				(7) edge (3)			
				(7) edge (5);	
			\end{tikzpicture}
			\caption{The graph $S$}
			\label{cirle_ob}
		\end{minipage}
		\end{figure}

	By extending the method given in \cite[Theorem 7]{glen18}, we now determine the representation number of $G^k$. As per the method, without loss of generality, we construct a word representing $G^k$ for an even number $k$. For $k \ge 5$, suppose $A = \{a, a_1,a_2,\ldots,a_k\}$ and  $B = \{b, b_1,b_2,\ldots,b_k\}$ such that $N(a) = B$, $N(b) = A$ and the $H_{k,k}$ is formed over the vertices $a_i$'s and $b_i$'s. Construct a word representing the crown graph $H_{k,k}$ by setting the permutations $P(a_r, a_s)$ over $a_i$'s and $P(b_r, b_s)$ over $b_i$'s, for all $ 1 \le r < s \le k$. Then, insert the vertex $a$ at the center in each permutation $P(a_r,a_s)$ and insert $b$ at the center in each permutation $P(b_r,b_s)$, for all $ 1 \le r < s \le k$. The resultant word represents $G^k$, as $a$ alternates with all vertices of $B$ and does not alternate with vertices of $A$. Similarly, $b$ alternates with all vertices of $A$ but does not alternate with vertices of $B$. Hence, $\mathcal{R}(G^k) = \lceil \frac{k}{2} \rceil$, for $k \ge 5$. Note that $\mathcal{R}(G^4) \ge 3$, as  $\mathcal{R}(H_{4,4}) = 3$. Hence, construct a word for $G^6$ and delete appropriate vertices from the word to observe that $\mathcal{R}(G^4) = 3$. Note that $\mathcal{R}(H_{3,3}) = 2$ but $G^3$ contains the graph $S$ given in Fig. \ref{cirle_ob} as induced subgraph. Since $S$ is a forbidden induced subgraph for circle graphs (cf. \cite{Bouchet_94}), we conclude that $\mathcal{R}(G^3) = 3$. For $k = 1, 2$, as $\mathcal{R}^p(G^k) = 2$, we have $\mathcal{R}(G^k) = 2$.
	
\end{example}

\section{Extended Crown Graphs}
\label{Example-ecg}

In this section, we provide a class of bipartite graphs, called extended crown graphs, for which the \textit{prn} can be investigated in view of Theorem \ref{main_result}. Thus, extended crowns graphs are a class of bipartite graphs whose \textit{prn} attains the upper bound given in this work. In this section, we also provide a partial characterization for extended crown graphs of \textit{prn} three.

An extended crown graph is constructed by extending a crown graph based on a given poset. Let $P$ be a poset. We construct a bipartite graph $G = (A \cup B, E)$ on $2n$ vertices, where $|P| = n$,  as per the following procedure: 

\begin{enumerate}
	\item Let $A$ be the set of elements of $P$ and $B = \{v'\; :\; v \in A\}$.
	\item Consider a largest antichain $M = \{a_1,a_2,\ldots, a_k\} \subseteq A$ of the poset $P$.  
	\item Let $M' = \{a_1', a_2', \ldots, a_k'\} \subseteq B$.
	\item For $1 \le i, j \le k$ and $i \ne j$, put $\{a_i, a_j'\} \in E$, i.e., construct a crown graph of $2k$ vertices with $M$ and $M'$ as partite sets.
	\item Then, for all $v \in A \setminus M$,  put $\{v, v'\} \in E$. 
	\item Further, for all $u, v \in A$, if $u < v$ in $P$, then iteratively make the vertices of $N(u)$ adjacent to $v$, i.e., for all $w \in N(u), \{w, v\} \in E$.
	\item The resultant graph $G = (A \cup B, E)$ is an extended crown graph. 
\end{enumerate}

\begin{remark}
	If a poset on $n$ elements is an antichain, then the corresponding extended crown graph is the crown graph $H_{n,n}$.
\end{remark}

\begin{remark}\label{w1poset}
	An extended crown graph of a chain will have $H_{1,1}$ as its largest induced crown graph. An example of the extended crown graph of a 4-element chain considering the antichain $\{1\}$ is depicted in Fig. \ref{ec_4-c}.  
	\begin{figure}
		\centering
		\begin{minipage}{.4\textwidth}
			\centering
			\begin{tikzpicture}[scale=0.5]
			\vertex (1) at (1,0) [label=left:$ 1$] {};  
			\vertex (2) at (1,-1) [label=left:$2$] {};
			\vertex (3) at (1,-2) [label=left:$3$] {};
			\vertex (4) at (1,-3) [label=left:$4$] {};
			\node (P) at (0,-1.5) [label=left:$P :$] {};
			\path		
			(1) edge (2)
			(2) edge (3)
			(3) edge (4);
			\end{tikzpicture}
		\end{minipage}
		\begin{minipage}{.4\textwidth}
			\centering
			\begin{tikzpicture}[scale=0.5]
			\vertex (1) at (1,0) [label=above:$ 1$] {};  		
			\vertex (2) at (2,0) [label=above:$2$] {};
			\vertex (3) at (3,0) [label=above:$3$] {};		
			\vertex (4) at (4,0) [label=above:$4 $] {};			
			\vertex (11) at (1,-1.5) [label=below:$ 1'$] {};  		
			\vertex (12) at (2,-1.5) [label=below:$2'$] {};
			\vertex (13) at (3,-1.5) [label=below:$3'$] {};		
			\vertex (14) at (4,-1.5) [label=below:$4'$] {};			
			\node (P) at (0.5,-0.75) [label=left:$G: $] {};
			\path	
			(2) edge (12)
			(3) edge (13)
			(4) edge (14)
			(1) edge (12)
			(1) edge (13)
			(1) edge (14)
			(2) edge (13)
			(2) edge (14)
			(3) edge (14);
			\end{tikzpicture}
		\end{minipage}
		\caption{An extended crown graph of a 4-element chain}
		\label{ec_4-c}
	\end{figure}
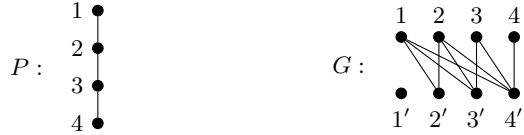
\end{remark}

\begin{remark}
	Extended crown graphs generated by a poset (using different largest antichains) need not be isomorphic. However, extended crown graphs generated by a poset of width two are isomorphic.
\end{remark}

\begin{remark}\label{re_w2}
	Let $G = (A \cup B, E)$ be the extended crown graph generated by a poset $P$ of width two. It can be easily observed that $P_A$ and $P$ are the same, where $P_A$ is the poset induced by the neighborhood graph of $G$ with respect to $A$. This need not be true for posets of width at least three. For example, the poset $P$ and the corresponding $P_A$ given in Fig. \ref{pa!=p} are not the same.
	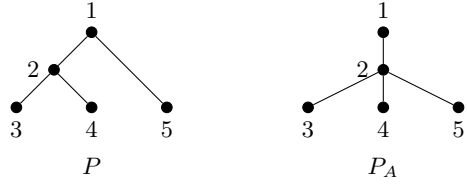
\begin{figure}
		\centering
		\begin{tabular}{ccc}
			\begin{tikzpicture}[scale=0.5]
			\vertex (1) at (0,0) [label=above:$1$] {};  		
			\vertex (2) at (-1,-1) [label=left:$2$] {};
			\vertex (3) at (-2,-2) [label=below:$3$] {};		
			\vertex (4) at (0,-2) [label=below:$4$] {};	
			\vertex (5) at (2,-2) [label=below:$5$] {};
			\path
			(1) edge (2)		
			(2) edge (3)
			(2) edge (4)
			(1) edge (5);			
			\end{tikzpicture}&\qquad \qquad \qquad&
			\begin{tikzpicture}[scale=0.5]
			\vertex (1) at (0,0) [label=above:$1$] {};  		
			\vertex (2) at (0,-1) [label=left:$2$] {};
			\vertex (3) at (-2,-2) [label=below:$3$] {};		
			\vertex (4) at (0,-2) [label=below:$4$] {};	
			\vertex (5) at (2,-2) [label=below:$5$] {};
			\path
			(1) edge (2)		
			(2) edge (3)
			(2) edge (4)
			(2) edge (5);			
			\end{tikzpicture}\\
			$P$ & & $P_A$
		\end{tabular}
		\caption{An example for Remark \ref{re_w2}}
		\label{pa!=p}
	\end{figure} 
\end{remark}

However, in the following proposition, we observe that the widths of $P$ and $P_A$ are the same. Accordingly, the \textit{prn} of extended crown graphs can be estimated as stated in Corollary \ref{prn_ec}.

\begin{proposition}\label{wPaisw_P}
	Suppose $G = (A \cup B, E)$ is an extended crown graph defined by a poset $P$. Then, $w(P_A) = w(P)$ and hence, $\kappa_0 = w(P)$.
\end{proposition}

\begin{proof}
	Let $M$ be a largest antichain of $P$. It can be observed from the construction of $G$ using $M$ that the neighborhoods of any pair of vertices in $M$ are incomparable, as the vertices of $M$ and $M'$ form a crown graph. Hence, $M$ is an antichain in $P_A$ so that $w(P_A) \ge w(P)$. 
	
	On the other hand, as per Step-6 in the construction of $G$, if $u < v$ in $P$, then $N(u) \subseteq N(v)$ in $G$. Thus, any two comparable elements in $P$ are comparable in $P_A$. Hence, $w(P_A) \le w(P)$. 
	
	Further, since $G$ has a crown graph of size $w(P)$, we have $w(P_B) \ge w(P_A)$. Consequently, $\kappa_0 = \min\{w(P_A), w(P_B)\} = w(P)$.
	\qed
\end{proof}

If $w(P) = k$, since $H_{k, k}$ is an induced subgraph of $G$, in view of Lemma \ref{lowerbd}, we have the following corollary of Proposition \ref{wPaisw_P}.

\begin{corollary} \label{crown_bound}
	If $w(P) = k$, then $H_{k,k}$ is the largest induced crown graph of $G$.
\end{corollary}

This, in turn, establishes that the extended crown graphs satisfy Theorem \ref{main_result}, as stated in the following corollary. 

\begin{corollary}\label{prn_ec}
	If $G$ is an extended crown graph, then $\mathcal{R}^p(G) = \kappa_0$ or $1 + \kappa_0 $.
\end{corollary}

\begin{remark}
	If $G$ is an extended crown graph defined by a poset of width one, then $\mathcal{R}^p(G) = 2$ because $G$ is a bipartite graph. 
\end{remark}	
In the following, we characterize the \textit{prn} of extended crown graphs defined by the posets of width two. 

\subsection{Posets of Width 2}
\label{poset_w2}
	
Let $P$ be a poset of width two. By Theorem \ref{main_result}, the \textit{prn} of an extended crown graph $G$ defined by $P$ is either two or three. Here, in Theorem \ref{chr_prn3}, we prove that $\mathcal{R}^p(G) = 2$ if and only if the cover graph of $P$ is a path or a disconnected graph. 

First, note that if the cover graph of $P$ is disconnected, then $P$ is a union of two chains (in which no two elements chosen one each from the two chains are comparable). In this case, the graph $G$ is disconnected, having two components. It can be observed that the poset induced by the neighborhood graph of each of these components is a chain. Accordingly, in view of Example \ref{ex_pa_chain}, we can conclude the \textit{prn} of $G$ as per the following remark.

\begin{remark}
	If $P$ is a union of two chains, then $\mathcal{R}^p(G) = 2$. 
\end{remark}

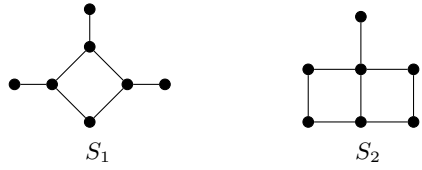
\begin{figure}
	\centering
	\begin{tabular}{ccc}
		\begin{tikzpicture}[scale=0.5]
		\vertex (1) at (1,1) [label=above:$ $] {};  
		\vertex (2) at (2,2) [label=left:$ $] {};
		\vertex (4) at (2,0) [label=right:$ $] {};
		\vertex (6) at (0,1) [label=right:$ $] {};
		\vertex (3) at (3,1) [label=left:$ $] {};
		\vertex (5) at (2,3) [label=right:$ $] {};
		\vertex (7) at (4,1) [label=right:$ $] {};
		\path		
		(1) edge (2)
		(2) edge (3)
		(3) edge (4)
		(1) edge (4)
		(2) edge (5)
		(1) edge (6)
		(3) edge (7); 		
		\end{tikzpicture}&\qquad \qquad\qquad&
		\begin{tikzpicture}[scale=0.7]
		\vertex (1) at (1,1) [label=above:$ $] {};  
		\vertex (2) at (2,1) [label=left:$ $] {};
		\vertex (3) at (2,0) [label=left:$ $] {};
		\vertex (4) at (1,0) [label=right:$ $] {};
		\vertex (5) at (3,1) [label=right:$ $] {};
		\vertex (6) at (3,0) [label=right:$ $] {};
		\vertex (7) at (2,2) [label=right:$ $] {};
		\path		
		(1) edge (2)
		(2) edge (3)
		(3) edge (4)
		(1) edge (4)
		(2) edge (5)
		(5) edge (6)
		(3) edge (6)
		(2) edge (7); 		
		\end{tikzpicture}\\
		$S_1$ & \qquad & $S_2$
	\end{tabular}	
	\caption{Forbidden induced subgraphs for \textit{prn} two graphs}
	\label{forb_perm}
\end{figure}

If the cover graph of $P$ is connected, then one among the posets $P_1, P_1', P_2$, $P_3, P_3', P_4$ and $P_5$ given in Fig. \ref{extended_crowns} is a subposet of $P$. Note that the cover graph of each of $P_1, P_1'$ and $P_2$ is a path, but it is not in the case of $P_3, P_3', P_4$ and $P_5$. The extended crown graph defined by both $P_3, P_3'$ is $G_3$, while it is $G_4$ for $P_4$ and $G_5$ for $P_5$ as shown in Fig. \ref{extended_crowns}. Recall from \cite{Gallai} that $S_1$ and $S_2$ given in Fig. \ref{forb_perm} are forbidden induced subgraphs for the \textit{prn} two graphs. As $S_1$ is an induced subgraphs of $G_3$ and $G_4$, and $S_2$ is an induced subgraph of $G_5$, the \textit{prn} of $G_3$, $G_4$ and $G_5$ are at least 3. By the method given in Section \ref{poly_proc}, we can construct 3-uniform words that permutationally represent $G_3, G_4$, and $G_5$ (see Appendix \ref{word_construct}). Hence, $\mathcal{R}^p(G_3) =\mathcal{R}^p(G_4) = \mathcal{R}^p(G_5) = 3$.

\begin{figure}[h!]
	\centering
	\begin{tabular}{|c|c|c|}
		\hline
		&&\\
		\begin{tikzpicture}[scale=0.5]
		\vertex (1) at (1,1) {};  
		\vertex (2) at (2,2) {};
		\vertex (3) at (3,1) {};
		\path
		(1) edge (2)		
		(2) edge (3);
		\end{tikzpicture}&
		\begin{tikzpicture}[scale=0.5]
		\vertex (1) at (1,2) {};  
		\vertex (2) at (2,1) {};
		\vertex (3) at (3,2) {};
		\path
		(1) edge (2)		
		(2) edge (3);
		\end{tikzpicture}&
		\begin{tikzpicture}[scale=0.6]
		\vertex (1) at (1,1) {};  
		\vertex (2) at (1,2) {};
		\vertex (3) at (2,1) {};
		\vertex (4) at (2,2) {};
		\path
		(1) edge (2)		
		(2) edge (3)
		(3) edge (4);
		\end{tikzpicture}\\
		$P_1$ & $P_1'$ & $P_2$\\	
		\hline
		\begin{tikzpicture}[scale=0.5]
		\vertex (1) at (0,0) [label=above:$ $] {};  		
		\vertex (2) at (0,-1) [label=left:$ $] {};
		\vertex (3) at (-1,-2) [label=below:$ $] {};		
		\vertex (4) at (1,-2) [label=below:$ $] {};	
		\path
		(1) edge (2)		
		(2) edge (3)
		(2) edge (4);			
		\end{tikzpicture} 
		\begin{tikzpicture}[scale=0.5]
		\vertex (1) at (0,0) [label=above:$ $] {};  		
		\vertex (2) at (0,1) [label=left:$ $] {};
		\vertex (3) at (-1,2) [label=below:$ $] {};		
		\vertex (4) at (1,2) [label=below:$ $] {};	
		\path
		(1) edge (2)		
		(2) edge (3)
		(2) edge (4);			
		\end{tikzpicture} & 
		\begin{tikzpicture}[scale=0.5]
		\vertex (1) at (-2,0) [label=above:$ $] {};  		
		\vertex (2) at (-2,-2) [label=below:$ $] {};
		\vertex (3) at (0,0) [label=above:$ $] {};		
		\vertex (4) at (0,-2) [label=below:$ $] {};	
		\path
		(1) edge (2)
		(1) edge (4)
		(3) edge (2)
		(3) edge (4);			
		\end{tikzpicture} & 	
		\begin{tikzpicture}[scale=0.5]
		\vertex (1) at (0,0) [label=above:$ $] {};  		
		\vertex (2) at (-1,-1) [label=left:$ $] {};
		\vertex (3) at (0,-2) [label=below:$ $] {};		
		\vertex (4) at (1,-1) [label=right:$ $] {};	
		\path
		(1) edge (2)		
		(1) edge (4)
		(2) edge (3)
		(4) edge (3);			
		\end{tikzpicture}	 \\
		$P_3$\qquad\qquad $P_3'$ & $P_4$ & $P_5$ \\
		\begin{tikzpicture}[scale=0.5]
		\vertex (1) at (0,-1) [label=left:$ $] {};  		
		\vertex (2) at (0,-3) [label=below:$ $] {};
		\vertex (3) at (-2,-2) [label=left:$ $] {};		
		\vertex (4) at (2,-2) [label=right:$ $] {};		
		\vertex (11) at (0,0) [label=above:$ $] {};  		
		\vertex (12) at (0,-2) [label=left:$ $] {};
		\vertex (13) at (1,-2) [label=above right:$ $] {};		
		\vertex (14) at (-1,-2) [label=above left:$ $] {};
		\node (P1) at (0,-4)	[label=below:$ $] {};		
		\path
		(11) edge (1)
		(1) edge (14)
		(1) edge (12)
		(1) edge (13)
		(2) edge (14)
		(2) edge (12)
		(2) edge (13)
		(3) edge (14)
		(4) edge (13);			
		\end{tikzpicture} & 
		\begin{tikzpicture}[scale=0.5]
		\vertex (1) at (0,-1) [label=left:$ $] {};  		
		\vertex (2) at (0,-3) [label=below:$ $] {};
		\vertex (3) at (-2,-2) [label=left:$ $] {};		
		\vertex (4) at (2,-2) [label=right:$ $] {};		
		\vertex (11) at (0,0) [label=above:$ $] {};  		
		\vertex (12) at (0,-4) [label=left:$ $] {};
		\vertex (13) at (1,-2) [label=above right:$ $] {};		
		\vertex (14) at (-1,-2) [label=above left:$ $] {};
		\path
		(11) edge (1)
		(1) edge (14)
		(1) edge (13)
		(2) edge (14)
		(2) edge (12)
		(2) edge (13)
		(3) edge (14)
		(4) edge (13);			
		\end{tikzpicture} & 
		\begin{tikzpicture}[scale=0.7]
		\vertex (1) at (0,0) [label=left:$ $] {};  		
		\vertex (2) at (1,0) [label=below:$ $] {};
		\vertex (3) at (2,0) [label=left:$ $] {};		
		\vertex (4) at (0,-1) [label=right:$ $] {};		
		\vertex (5) at (1,-1) [label=above:$ $] {};  		
		\vertex (6) at (2,-1) [label=left:$ $] {};
		\vertex (7) at (1,1) [label=above right:$ $] {};		
		\vertex (8) at (1,-2) [label=above left:$ $] {};	
		\path
		(1) edge (2)
		(2) edge (3)
		(3) edge (6)
		(6) edge (5)
		(5) edge (4)
		(4) edge (1)
		(2) edge (5)
		(7) edge (2)
		(8) edge (5);			
		\end{tikzpicture}\\
		$G_3$ & $G_4$ & $G_5$ \\
		\hline
	\end{tabular}
	\caption{Posets of width two and corresponding extended crown graphs}
	\label{extended_crowns}			
\end{figure}
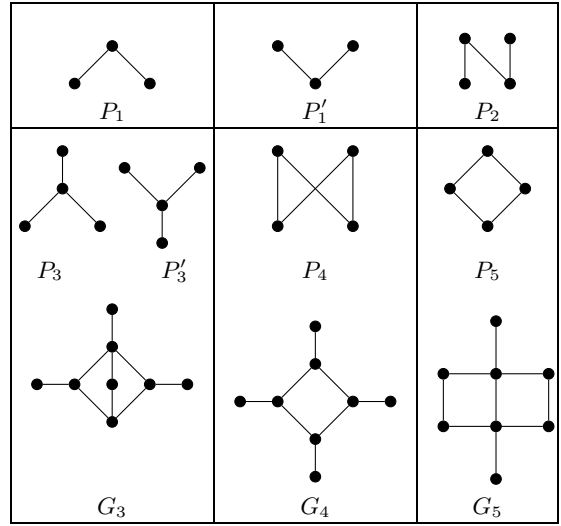

Suppose the cover graph of $P$ is a path. Then, $P$ has $P_1, P_1'$ or $P_2$ as one of its subposets, and hence, $P$ can be seen as any one of the two types shown in Fig. \ref{path_poset}. Accordingly, as per the following cases, we prove that an extended crown graph $G$ defined by $P$ has \textit{prn} two. 
\begin{itemize}
	\item $P$ is of Type-1: Let $P = X_1 \cup X_2 \cup \{a\}$, where $X_1$ and $X_2$ are chains and none of the elements of $X_1$ are comparable to elements of $X_2$ and vice versa. Also, the element $a$ is comparable to all elements of $P$. Note that $\{\{a\} \cup X_1, X_2\}$ is a minimal chain cover of $P$ and hence of $P_A$ in view of Remark \ref{re_w2}. The chain cover satisfies the condition given in Theorem \ref{k-rep}. Hence, $\mathcal{R}^p(G) = 2$.
	\begin{figure}
		\centering
		\begin{minipage}{.5\textwidth}
			\centering
			\begin{tikzpicture}[scale=0.5]
			\vertex (1) at (1,1) [label=left:$ $] {};  
			\vertex (2) at (2,2) [label=below:$ a $] {};
			\vertex (3) at (3,1) [label=right:$ $] {};
			\vertex (4) at (3,0) [label=below:$ $] {};
			\vertex (5) at (1,0) [label=below:$ $] {};
			\vertex (6) at (1,-1) [label=left:$ $] {}; 
			\vertex (7) at (3,-1) [label=below:$ $] {};
			\vertex (8) at (1,-2) [label=left:$ $] {}; 
			\vertex (9) at (3,-2) [label=below:$ $] {};
			\vertex (10) at (5,-1) [label=left:$ $] {};  
			\vertex (11) at (6,-2) [label=above:$ a $] {};
			\vertex (12) at (7,-1) [label=right:$ $] {};
			\vertex (13) at (7,0) [label=below:$  $] {};
			\vertex (14) at (5,0) [label=below:$  $] {};
			\vertex (15) at (5,1) [label=left:$  $] {}; 
			\vertex (16) at (7,1) [label=below:$  $] {};
			\vertex (17) at (5,2) [label=left:$ $] {}; 
			\vertex (18) at (7,2) [label=below:$ $] {};
			\path		
			(1) edge (2)
			(2) edge (3)
			(3) edge (4)
			(5) edge (1)
			(6) edge (8)
			(7) edge (9)
			(6) edge [dashed] (5)
			(7) edge[dashed] (4)
			(10) edge (11)
			(11) edge (12)
			(12) edge (13)
			(14) edge (10)
			(15) edge (17)
			(16) edge (18)
			(15) edge [dashed] (14)
			(16) edge[dashed] (13);
			\end{tikzpicture}
			
			Type-1
		\end{minipage}
		\begin{minipage}{.3\textwidth}
			\centering
			\begin{tikzpicture}[scale=0.5]
			\vertex (1) at (1,1) [label=left:$a_k $] {};  
			\vertex (3) at (3,1) [label=right:$a_{n} $] {};
			\vertex (4) at (3,0) [label=right:$a_{n-1} $] {};
			\vertex (5) at (1,0) [label=left:$a_{k-1} $] {};
			\vertex (14) at (3,-1) [label=right:$a_j $] {};
			\vertex (15) at (1,-1) [label=left:$a_i $] {};
			\vertex (6) at (1,-2) [label=left:$ a_2$] {}; 
			\vertex (7) at (3,-2) [label=right:$a_{k+2} $] {};
			\vertex (8) at (1,-3) [label=left:$ a_1$] {}; 
			\vertex (9) at (3,-3) [label=right:$ a_{k+1} $] {};
			\path		
			(3) edge (4)
			(5) edge (1)
			(6) edge (8)
			(7) edge (9)
			(1) edge (9)
			(4) edge [dashed] (14)
			(5) edge[dashed] (15)
			(15) edge [dashed] (6)
			(7) edge[dashed] (14);
			\end{tikzpicture}
			
			Type-2
		\end{minipage}
		\caption{Posets whose cover graph is a path}
		\label{path_poset}
	\end{figure}
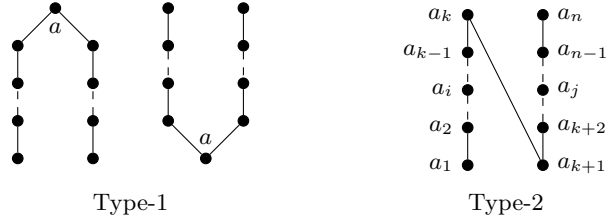		
	
	\item $P$ is of Type-2: Consider the largest antichain $\{a_i, a_j\}$ of $P$. The respective extended crown graph $G$ is as shown in Fig. \ref{ec_type-2}.  Consider the following permutations :
	$$p_1 : w_1 a_{k+1} a_k a_{k+1}' a_k' w_2 \ \text{and} \ p_2 : w_3 w_4 $$		
	where, 
	\begin{align*}
	w_1 &= a_n a_n' a_{n-1} a_{n-1}'\cdots a_j a_i' \cdots a_{k+2} a_{k+2}'\\
	w_2 &= a_{k-1} a_{k-1}'\cdots a_i a_j' \cdots a_{2} a_{2}' a_{1} a_{1}' \\
	w_3 &= a_{1} a_{2} \cdots a_i \cdots a_k a_{1}' a_{2}' \cdots a_j' \cdots a_{k}'\\
	w_4 &= a_{k+1} a_{k+2} \cdots a_j\cdots a_n a_{k+1}' a_{k+2}' \cdots a_i'\cdots a_n'		
	\end{align*}
	It can be observed that the word $p_1p_2$ represents $G$ permutationally. Hence, $\mathcal{R}^p(G) = 2$.
\end{itemize}

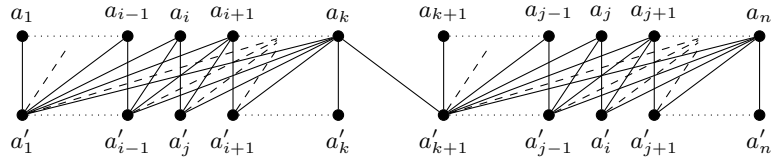
\begin{figure}
	\centering
	\begin{tikzpicture}[scale=0.7]
	\vertex (1) at (0,0) [label=above:$a_1$] {};
	\node (a) at (1,0) [label=above:$$] {};
	\vertex (2) at (2,0) [label=above:$a_{i-1}$] {};
	\vertex (3) at (3,0) [label=above:$a_i$] {};		
	\vertex (4) at (4,0) [label=above:$a_{i+1}$] {};
	\node (b) at (5,0) [label=above:$$] {};
	\vertex (5) at (6,0) [label=above:$a_{k}$] {};	
	\vertex (6) at (8,0) [label=above:$a_{k+1}$] {};	
	\node (c) at (9,0) [label=above:$$] {};
	\vertex (7) at (10,0) [label=above:$a_{j-1}$] {};	
	\vertex (8) at (11,0) [label=above:$a_j$] {};	
	\vertex (9) at (12,0) [label=above:$a_{j+1}$] {};	
	\node (d) at (13,0) [label=above:$$] {};
	\vertex (10) at (14,0) [label=above:$a_n$] {};	
	
	\vertex (11) at (0,-1.5) [label=below:$a_1'$] {};
	\vertex (12) at (2,-1.5) [label=below:$a_{i-1}'$] {};
	\vertex (13) at (3,-1.5) [label=below:$a_j'$] {};		
	\vertex (14) at (4,-1.5) [label=below:$a_{i+1}'$] {};	
	\vertex (15) at (6,-1.5) [label=below:$a_{k}'$] {};	
	\vertex (16) at (8,-1.5) [label=below:$a_{k+1}'$] {};	
	\vertex (17) at (10,-1.5) [label=below:$a_{j-1}'$] {};	
	\vertex (18) at (11,-1.5) [label=below:$a_i'$] {};	
	\vertex (19) at (12,-1.5) [label=below:$a_{j+1}'$] {};	
	\vertex (20) at (14,-1.5) [label=below:$a_n'$] {};	
	\path		
	(1) edge[dotted] (2)
	(4) edge[dotted] (5)
	(6) edge[dotted] (7)
	(9) edge[dotted] (10)
	(11) edge[dotted] (12)
	(14) edge[dotted] (15)
	(16) edge[dotted] (17)
	(19) edge[dotted] (20)
	(11) edge (1)
	(11) edge (2)
	(11) edge (3)
	(11) edge (4)
	(11) edge (5)
	(11) edge[dashed] (a)
	(11) edge[dashed] (b)
	(12) edge (2)
	(12) edge (3)
	(12) edge (4)
	(12) edge (5)	
	(12) edge[dashed] (b)
	(13) edge (3)
	(13) edge (4)
	(13) edge (5)	
	(13) edge[dashed] (b)
	(14) edge (4)
	(14) edge (5)	
	(14) edge[dashed] (b)
	(15) edge (5)
	(16) edge (5)
	(16) edge (6)
	(16) edge (7)
	(16) edge (8)
	(16) edge (9)
	(16) edge (10)
	(16) edge[dashed] (c)
	(16) edge[dashed] (d)
	(17) edge (7)
	(17) edge (8)
	(17) edge (9)
	(17) edge (10)	
	(17) edge[dashed] (d)
	(18) edge (8)
	(18) edge (9)
	(18) edge (10)	
	(18) edge[dashed] (d)
	(19) edge (9)
	(19) edge (10)	
	(19) edge[dashed] (d)
	(20) edge (10);
	
	\end{tikzpicture}
	\caption{Extended crown graph for Type-2 poset}
	\label{ec_type-2}
\end{figure}

Thus, the result obtained in this section is summarized in the following theorem.

\begin{theorem}\label{chr_prn3}
	Let $G$ be the extended crown graph defined by a poset $P$ of width two. Then, $\mathcal{R}^p(G) = 3$ if and only if the cover graph of $P$ is connected and not a path.
\end{theorem}

In the following, we further provide the representation number of a special class of extended crown graphs.

\begin{theorem}\label{ec_cycle}
	Let $P$ be a poset of width two whose cover graph is a cycle. If $G$ is an extended crown graph defined by $P$, then $\mathcal{R}(G) = 2$.
\end{theorem}

\begin{proof}
	Since $G$ is a bipartite graph with more than two vertices, $R(G) \ge 2$. Since $w(P)$ = 2, by Theorem \ref{chr_prn3}, $R(G) \le 3$. If $P$ is of height two, then  $P_4$ is the only possible poset, and $\mathcal{R}(G_4) = 2$, as $G_4$ is a circle graph. If $P$ is of height greater than two, then  we show that there is a 2-uniform word that represents $G$.
	
	Let $\{X_1, X_2\}$ be a minimal chain cover of $P$ such that $X_1 : a_1 < a_2 < \cdots < a_{r}$ and $X_2 :  b_1 < b_2 < \cdots< b_{s}$. Assume $s < r$ so that $a_1$ and $a_{r}$ are the minimum and maximum elements of $P$, respectively.  If $\{a_i,b_j\}$, for some $i$ and $j$, is the largest antichain chosen.\\
	Consider the following words:
	\begin{align*}
	w_{1} & = a_{r} a_{r}' a_{r -1 } a_{r -1}' \cdots  a_i b_j' \cdots  a_2 a_2' a_1 a_1', \text{with the elements of } X_1 \cup X_1'\\
	w_{2}& = a_1 a_2 \cdots a_{r -1}, \text{with the elements of } X_1 \setminus \{a_{r}\} \\
	w_{3} & = b_1' b_2' \cdots a_i'\cdots b_{s}', \text{with the elements of } X_2'\\
	w_{4} &= b_{s} b_{s}' b_{s -1 } b_{s -1}'  \cdots  b_j a_i'\cdots b_2 b_2' b_1 b_1', \text{with the elements of } X_2 \cup X_2'\\
	w_{5} & = b_1 b_2 \cdots b_{s}, \text{with the elements of } X_2\\
	w_{6} & =  a_2' \cdots b_j'\cdots a_{r}', \text{with the elements of }  X_1'
	\end{align*}
	Note that the word $w$ is given by
	\[w = w_{1} w_{2} w_{3} a_{r} w_{4} a_1' w_{5} w_{6}\]
	is a 2-uniform word. It can be observed that any two vertices $a$ and $b$ are adjacent in $G$ if and only if $a$ and $b$ alternate in $w$ (cf. Appendix \ref{proof_details}). Hence, $\mathcal{R}(G) = 2$.\qed 
\end{proof}

\begin{corollary}
	Let $P$ be a poset of width two whose cover graph is a cycle. An extended crown graph defined by $P$ is a circle graph.
\end{corollary}

\section{Conclusion}

In this paper, we observed that the notions of the \textit{prn} of a comparability graph and the dimension of its induced poset are the same, and accordingly, we reconciled the \textit{prn} of some comparability graphs. We have employed the notion of a neighborhood graph to provide a cubic-time computational procedure for constructing permutations of the vertices of a given bipartite graph whose concatenation represents the bipartite graph permutationally. As a result, we establish a better upper bound for the \textit{prn} of a bipartite graph compared to the best-known upper bound given in \cite{Broere18}. In fact, the given procedure estimates the \textit{prn} of certain bipartite graphs. Further, based on our observations, the representation number of a specific type of bipartite graph is conjectured as per the following:  

\begin{conjecture}
	Let $G$ be a bipartite graph. For $\kappa_0 \ge 5$, if $H_{\kappa_0,\kappa_0}$ is a largest induced crown graph in $G$, then $\mathcal{R}(G)$ equals $\lceil \frac{\kappa_0}{2} \rceil$ or $\lceil \frac{1+\kappa_0}{2} \rceil$.
\end{conjecture}

This works opens up some specific questions towards finding the \textit{prn} of various subclasses of comparability graphs. For instance, one may target to characterize the graph classes that correspond to the poset classes whose dimensions are known in the literature. Further, it is interesting to extend the given computational procedure to comparability graphs in general. In line with the contributions in this paper, we also ask some specific questions related to the class of extended crown graphs: 1) Whether an extended crown graph defined by a poset of width two is a circle graph? 2) Characterize the posets of width at least three in terms of the \textit{prn} (and representation number) of their extended crown graphs.

\subsubsection{Acknowledgment.} 
The first author is thankful to the Council of Scientific and Industrial Research (CSIR), Government of India, for awarding the research fellowship for pursuing Ph.D. at IIT Guwahati. 

\bibliographystyle{abbrv}

\appendix

\section{Demonstration of Procedure given in Section \ref{poly_proc}}
\label{demo}

\begin{example}\label{chain_ex}	
	Let $G = (A \cup B, E)$ be the bipartite graph given in Fig. \ref{fig_G1}.	
	\begin{figure}[!htb]
		\centering
		\begin{minipage}{.55\textwidth}
			\centering
			\begin{tikzpicture}[scale=0.5]				
				\node (A) at (-0.2,0) [label=left:$A$] {};				
				\node (B) at (-0.2,2) [label=left:$B$] {};			
				\vertex (1) at (0,0) [label=below:$1$] {};  
				\vertex (2) at (1,0) [label=below:$2$] {};
				\vertex (3) at (2,0) [label=below:$3$] {};
				\vertex (4) at (3,0) [label=below:$4$] {};
				\vertex (5) at (4,0) [label=below:$5$] {};  
				\vertex (6) at (5,0) [label=below:$6$] {};
				\vertex (7) at (0,2) [label=above:$7$] {};
				\vertex (8) at (1,2) [label=above:$8$] {};
				\vertex (9) at (2,2) [label=above:$9$] {};
				\vertex (10) at (3,2) [label=above:$10$] {};
				\vertex (11) at (4,2) [label=above:$11$] {};  				
				\path				
				(1) edge (8)
				(1) edge (9)
				(2) edge (9)
				(2) edge (10)
				(3) edge (11)				
				(4) edge (7)
				(4) edge (8)
				(4) edge (9)
				(4) edge (10)
				(5) edge (8)
				(5) edge (9)
				(5) edge (10)
				(5) edge (11)
				(6) edge (7)
				(6) edge (8)
				(6) edge (9)	
				(6) edge (10)
				(6) edge (11);		
			\end{tikzpicture}
			\caption{A bipartite graph}
			\label{fig_G1}
		\end{minipage}
		\begin{minipage}{.4\textwidth}
			\centering
			\begin{tikzpicture}[scale=0.7]
				\vertex (1) at (0,0) [label=above:$ $] {};  
				\vertex (2) at (-1,-1) [label=left:$ $] {}; 
				\vertex (3) at (0,-1) [label=right:$ $] {}; 
				\vertex (4) at (1,-1) [label=right:$ $] {}; 
				\vertex (5) at (-2,-2) [label=below:$ $] {}; 
				\vertex (6) at (0,-2) [label=below:$ $] {}; 
				\vertex (7) at (2,-2) [label=below:$ $] {}; 
				\path
				(1) edge (2)
				(1) edge (3)
				(1) edge (4)
				(2) edge (5)
				(3) edge (6)
				(4) edge (7);
			\end{tikzpicture}
			\caption{$S_3$}
			\label{S_3}
		\end{minipage}		
	\end{figure}
	Note that the poset $P_{A}$ shown in Fig. \ref{fig_N1} has width three. Consider the following chain cover of $P_A$: 
	\begin{align*}
		X_1 &: 1 < 4\\
		X_2 &: 2 < 5 < 6 \\
		X_3 & : 3
	\end{align*}
	The corresponding permutations are as per the following, in which two-digit numbers (i.e., 10 and 11) are indicated in brackets. 
	\begin{align*}
		p_1 &= 2 5 6 3 (11) 4 7 (10) 1 8 9\\
		p_2 &= 1 4 3 6 7 5 8 (11) 2 9 (10)\\
		p_3 &= 1 4 2 5 6 9 8 (10) 7 3 (11)
	\end{align*}
	Note that the chain cover $\{X_1, X_2, X_3\}$ satisfies Corollary \ref{kappa0_rep} and hence, $p_1p_2p_3$ represent the graph $G$ permutationally.
	
	\begin{figure}[!htb]
		\centering
		\begin{minipage}{.4\textwidth}
			\centering
			\[\begin{tikzpicture}[scale=0.5]			
				\vertex (1) at (0,1) [label=below:$1$] {};  
				\vertex (2) at (2,1) [label=below:$2$] {};
				\vertex (3) at (4,1) [label=below:$3$] {};
				\vertex (4) at (1,2.5) [label=left:$4$] {};
				\vertex (5) at (3,2.5) [label=right:$5$] {};  
				\vertex (6) at (2,4) [label=above:$6$] {};
				\path				
				(1) edge (4)
				(1) edge (5)
				(2) edge (4)
				(2) edge (5)
				(3) edge (5)
				(4) edge (6)
				(5) edge (6);		
			\end{tikzpicture}\]
			\caption{$P_A$}
			\label{fig_N1}
		\end{minipage}
		\begin{minipage}{.4\textwidth}
			\centering
			\[\begin{tikzpicture}[scale=0.5]
				\vertex (7) at (0,1) [label=below:$7$] {};  
				\vertex (8) at (-1,2.5) [label=left:$8$] {};
				\vertex (9) at (0,4) [label=above:$9$] {};
				\vertex (10) at (1,2.5) [label=right:$10$] {};
				\vertex (11) at (3,2.5) [label=right:$11$] {};
				\path				
				(7) edge (8)
				(7) edge (10)
				(8) edge (9)
				(10) edge (9);				
			\end{tikzpicture}\]
			\caption{$P_B$}
			\label{fig_N3}
		\end{minipage}
	\end{figure}
	
	The poset $P_{B}$ shown in Fig. \ref{fig_N3}, has also width three. Consider the following chain cover of $P_B$: 
	\begin{align*}
		Y_1 &: 7<8<9 \\
		Y_2 &: 10 \\
		Y_3 &: 11 
	\end{align*}	
	The corresponding permutations are as per the following:
	\begin{align*}
		p_1' &= (10) (11) 3 9 2 8 5 1 7 6 4\\
		p_2' &=   7 8 9 (11) 1 3 (10)4 6 5 2 \\
		p_3' &= 7 8 9 (10) 4 1 2 (11) 6 5 3.
	\end{align*}
	The word $w' = p_1'p_2'p_3$  does not represent $G$  as the vertices 1 and 6 are not adjacent in $G$ while $w'_{\{1,6\}}=161616$. Also, note that
	\begin{align*}
		N_{\{7<8<9\}}(1) &\subsetneq N_{\{7<8<9\}}(6)\\
		N_{\{10\}}(1) &\subsetneq N_{\{10\}}(6)\\
		N_{\{11\}}(1) &\subsetneq N_{\{11\}}(6)
	\end{align*}
	We concatenate $p_0'=(11) (10) 7 8 9  4 6 1 5 2 3$ to the word $w'$ so that $p_0'p_1'p_2'p_3'$ represents $G$ permutationally. Hence, $\mathcal{R}^p(G) \le 3$. However, since $G$ contains $S_3$ shown in Fig. \ref{S_3}, which is a forbidden induced subgraph for permutation graphs \cite{Gallai}, we have $\mathcal{R}^p(G) = 3$.     
\end{example}

\section{Words for $G_3$, $G_4$ and $G_5$ given in Section \ref{poset_w2}}
\label{word_construct}

We use the procedure presented in Section \ref{poly_proc} and construct words which represent the graphs $G_3$, $G_4$ and $G_5$, given in Section \ref{poset_w2}, permutationally. Consider the labeling for the graphs as shown in Fig. \ref{labelling_g3-5}. 

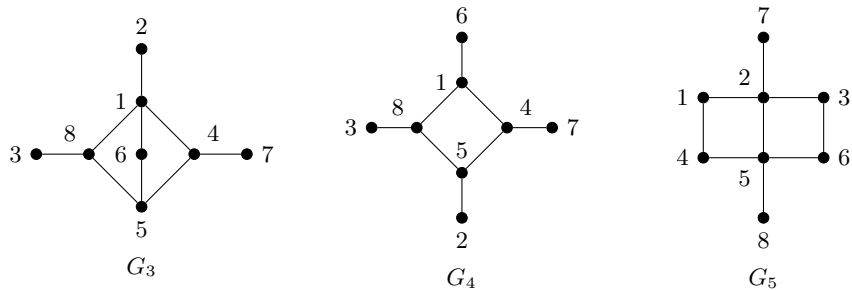
\begin{figure}[h!]
	\centering
	\begin{minipage}{.34\textwidth}
		\centering
		\begin{tikzpicture}[scale=0.7]
			\vertex (1) at (0,-1) [label=left:$ 1$] {};  		
			\vertex (2) at (0,-3) [label=below:$ 5$] {};
			\vertex (3) at (-2,-2) [label=left:$ 3$] {};		
			\vertex (4) at (2,-2) [label=right:$ 7$] {};		
			\vertex (11) at (0,0) [label=above:$ 2$] {};  		
			\vertex (12) at (0,-2) [label=left:$ 6$] {};
			\vertex (13) at (1,-2) [label=above right:$ 4$] {};		
			\vertex (14) at (-1,-2) [label=above left:$ 8$] {};		
			\path
			(11) edge (1)
			(1) edge (14)
			(1) edge (12)
			(1) edge (13)
			(2) edge (14)
			(2) edge (12)
			(2) edge (13)
			(3) edge (14)
			(4) edge (13);			
		\end{tikzpicture}
		
		$G_3$
	\end{minipage}
	\begin{minipage}{.34\textwidth} 
		\centering
		\begin{tikzpicture}[scale=0.6]
			\vertex (1) at (0,-1) [label=left:$1 $] {};  		
			\vertex (2) at (0,-3) [label=above:$ 5$] {};
			\vertex (3) at (-2,-2) [label=left:$ 3$] {};		
			\vertex (4) at (2,-2) [label=right:$ 7$] {};		
			\vertex (11) at (0,0) [label=above:$ 6$] {};  		
			\vertex (12) at (0,-4) [label=below:$2 $] {};
			\vertex (13) at (1,-2) [label=above right:$4 $] {};		
			\vertex (14) at (-1,-2) [label=above left:$ 8$] {};
			\path
			(11) edge (1)
			(1) edge (14)
			(1) edge (13)
			(2) edge (14)
			(2) edge (12)
			(2) edge (13)
			(3) edge (14)
			(4) edge (13);			
		\end{tikzpicture} 
		
		$G_4$
	\end{minipage}
	\begin{minipage}{.3\textwidth}
		\centering
		\begin{tikzpicture}[scale=0.8]
			\vertex (1) at (0,0) [label=left:$ 1$] {};  		
			\vertex (2) at (1,0) [label=above left:$ 2$] {};
			\vertex (3) at (2,0) [label=right:$3 $] {};		
			\vertex (4) at (0,-1) [label=left:$4 $] {};		
			\vertex (5) at (1,-1) [label=below left:$5 $] {};  		
			\vertex (6) at (2,-1) [label=right:$6 $] {};
			\vertex (7) at (1,1) [label=above:$7 $] {};		
			\vertex (8) at (1,-2) [label=below:$ 8$] {};	
			\path
			(1) edge (2)
			(2) edge (3)
			(3) edge (6)
			(6) edge (5)
			(5) edge (4)
			(4) edge (1)
			(2) edge (5)
			(7) edge (2)
			(8) edge (5);			
		\end{tikzpicture}
		
		$G_5$
	\end{minipage}
	\caption{A labeling of $G_3, G_4$ and $G_5$}
	\label{labelling_g3-5}
\end{figure}

\begin{figure}[h!]
	\centering
	\begin{minipage}{.3\textwidth}
		\centering
		\begin{tikzpicture}[scale=0.6]
		\vertex (1) at (0,0) [label=above:$ 1$] {};  		
		\vertex (5) at (0,-1) [label=below:$ 5$] {};
		\vertex (3) at (-1,-2) [label=below:$ 3$] {};		
		\vertex (7) at (1,-2) [label=below:$ 7$] {};		
		\path
		(1) edge (5)
		(5) edge (3)
		(5) edge (7);			
		\end{tikzpicture}
		
		$P_A$ of $G_3$
	\end{minipage}
	\begin{minipage}{.3\textwidth} 
		\centering
		\begin{tikzpicture}[scale=1.2]
		\vertex (1) at (0,0) [label=above:$1 $] {};  		
		\vertex (5) at (1,0) [label=above:$ 5$] {};
		\vertex (3) at (0,-1) [label=below:$ 3$] {};		
		\vertex (7) at (1,-1) [label=below:$ 7$] {};	
		\path
		(1) edge (3)
		(1) edge (7)
		(5) edge (3)
		(5) edge (7);			
		\end{tikzpicture} 
		
		$P_A$ of $G_4$
	\end{minipage}
	\begin{minipage}{.3\textwidth}
		\centering
		\begin{tikzpicture}[scale=0.7]
		\vertex (5) at (0,0) [label=above:$ 5$] {};  		
		\vertex (1) at (-1,-1) [label=left:$ 1$] {};
		\vertex (3) at (1,-1) [label=right:$3 $] {};		
		\vertex (7) at (0,-2) [label=below:$7 $] {};	
		\path
		(5) edge (1)
		(5) edge (3)
		(1) edge (7)
		(3) edge (7);			
		\end{tikzpicture}
		
		$P_A$ of $G_5$
	\end{minipage}
	\caption{Posets with respect to the partite set $A= \{1,3,5,7\}$}
	\label{p_as_g3-5}
\end{figure}
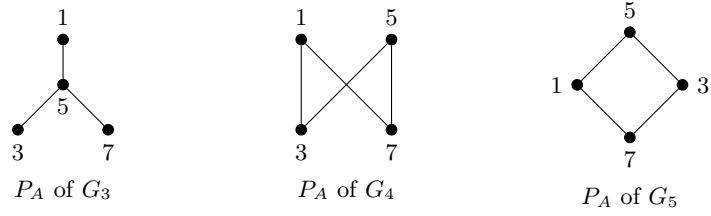

Note that $A = \{1,3,5,7\}$ and $B = \{2,4,6,8\}$ are the partite sets for each graph. In each case, the poset $P_A$ induced by the neighborhood graph with respect to $A$ is depicted in Fig. \ref{p_as_g3-5}. The words representing the graphs are obtained in the Table \ref{table_g3-5}. 

\begin{table}[h!]
	\centering
\begin{tabular}{|l|c|c|c|c|}
	\cline{3-5}
	\multicolumn{2}{l|}{}& $G_3$ & $G_4$ & $G_5$ \\
	\hline
	\multicolumn{2}{|l|}{Chain Cover} & $ \{3 < 5 < 1\}$ and $\{7\}$   & $ \{3 < 1\}$ and $\{7 < 5\}$  & $\{ 7 < 1 <  5\}$ and $\{3\}$	\\
	\hline
	\multirow{3}{*}{Permutations} & $p_1$ & 7 1 2 5 4 6 3 8 &  7 5 2 1 4 6 3 8 & 3 5 6 8 1 4 7 2 \\ \cline{2-5}
				& $p_2$ & 3 5 1 8 6 2 7 4 & 3 1 6 5 8 2 7 4 & 7 1 5 4 8 3 2 6 \\ \cline{2-5}
				& $p_0$ & 7 3 5 1 8 6 4 2 & 7 5 3 1 8 6 4 2 & 3 7 1 5 2 4 8 6 \\ \cline{2-5}
	\hline
	\multicolumn{2}{|l|}{Word} & $p_0p_1p_2$ & $p_0p_1p_2$ &	$p_0p_1p_2$ \\
	\hline
\end{tabular}\\[5pt]
\caption{Construction of words for $G_3$, $G_4$ and $G_5$}
\label{table_g3-5}
\end{table}

\section{Details for the Proof of Theorem \ref{ec_cycle}}
\label{proof_details}

\begin{proof}	
	Now, we show that $a$ and $b$ are adjacent in $G$ if and only if $a$ and $b$ alternate in $w$.
	Suppose $a$ and $b$ are not adjacent in $G$. We show that $a$ and $b$ do not alternate in $w$ through various cases as per the following:
	\begin{itemize}
		\item $a, b \in A$: 
		\begin{itemize}
			\item $a, b \in X_1$: Let $a= a_l$ and $b = a_m$, where $l < m$. We have, $baab \le  w_1 w_2 a_{r} \le w$.
			\item $a, b \in X_2$: Let $a = b_l$ and $b = b_m$, where $l < m$. Then,  $baab \le  w_4 w_5 \le w $.
			\item $a \in X_1, b \in X_2$: Let $a = a_l$ and $b = b_m$, for all $l$ and $m$. Then, $aabb \le  w_1  w_2 a_{r}   w_4 w_5 \le w$.
		\end{itemize} 
		
		\item $a, b \in B$: 
		\begin{itemize}
			\item $a, b \in B \setminus \{a_i',b_j'\}$: Let $a= a_l'$ and $b=a_m'$, where $l < m$. We have, $baab \le  w_1 a_1'w_6 \le w$. Similarly, $a= b_l'$ and $b=b_m'$, where $l < m$, then  $abba \le  w_3  w_4  \le w $. Assume $a = a_l'$ and $b = b_m'$, for all $l$ and $m$. Then, $abba \le  w_1 w_3   w_4 a_1' w_6 \le w$.
			\item $a = a_i'$: If $b = a_m'$, where $m \ne i$, then $baab \le w_1 w_3 w_4 w_6 \le w$. Let $b = b_m'$, where $m \ne j$. If $m < j$, then $baab \le w_3w_4 \le w$. Else, if $m > j$, then $abba \le w_3w_4 \le w$. If $b = b_j'$, then $baab \le w_1w_3w_4w_6 \le w.$
			\item $a = b_j'$: If $b = b_m'$, where $m \ne j$, then $abba \le w_1w_3w_4w_6 \le w$. Let $b = a_m'$, where $m \ne i$. If $m < i$, then $abba \le w_1 a_1' w_6 \le w$. Else, if $m >i$, then $baab \le w_1 w_6 \le w$. 
		\end{itemize}
		
		\item $a \in A$ and $b \in B$: 
		\begin{itemize}
			\item $a \in X_1 $: Let $a = a_l$, where $ 1 \le l \le {r-1}$. If $b = a_m'$, for all $m > l$. We have, $baab \le  w_1 w_2 w_6 \le w$. Similarly, $b = b_m'$, for all $m \ne j$. We have, $aabb \le  w_1 w_2 w_3  w_4  \le w $.
			\item $a \in X_2 $: Let $a= b_l$, where $1 \le l \le s$. If $b= b_m '$, for all  $m > l$. We have, $bbaa \le  w_3 w_4 w_5 \le w $. If $b= a_m'$, for all  $m \ne i$. We have $baab \le w_1  w_4  w_5 w_6 \le w$.
			\item $b = a_i'$: If $a = a_l$, for all $l$, then $aabb \le w_1 w_2 w_3 w_4 \le w$. If $a = b_m$, for all $m < j$, then $bbaa \le w_3w_4w_5 \le w$.
			\item $b = b_j'$: If $a = a_l$, for all $ l < i$, then $baab \le w_1 w_4 w_5 w_6 \le w$. If $a = b_m$, for all $m$, then $baab \le w_1 w_4 w_5 w_6 \le w$.
		\end{itemize} 
	\end{itemize}
	Thus, $a$ and $b$ do not alternate in $w$.
	
	Conversely, suppose $a$ and $b$ are adjacent in $G$. We deal this part in the following cases:
	\begin{itemize}
		\item  $a \in X_1$ and $b \in X_1'\setminus \{a_i'\}$: Let $a = a_l$ and $b = a_m'$, where $l > m $. We have, $abab \le w_1 w_2 a_{r}a_1'  w_6 \le w$.
		\item  $a \in X_2$ and $b \in X_2' \setminus \{b_j'\}$: Let $a = b_l$ and $b = b_m'$, where $l > m $. We have, $baba \le  w_3  w_4  w_5 \le w$ .
		\item $a = a_{r}$ and $b \in B$: Clearly, $abab \le  w_1  w_3 a_{r}  w_4 a_1' w_6\le  w$.  
		\item $a \in A$ and $b = a_1'$ : If $a \in X_1$, then $abab \le w_1 w_2 a_1' \le w $. Else, if $a \in X_2$, then $baba \le w_1 w_4 a_1'w_5 \le w.$ 
		\item $a \in X_1$ and $b = b_j'$: Let $a = a_l$, where $l \ge  i$. Then, $abab \le w_1w_2w_6 \le w$.
		\item  $a \in X_2$ and $b = a_i'$: Let $a = b_l$, where $l \ge  j$. Then, $baba \le w_3 w_4 w_5 \le w$.
	\end{itemize}
	Thus, $a$ and $b$ alternate in $w$.
	\qed
\end{proof}

\end{document}